\pdfoutput=1
\documentclass{article}

\usepackage[final]{neurips_2018}

\usepackage[usenames,dvipsnames,dvinames]{xcolor}
\usepackage{url,times,stfloats,flushend,float,crop,color,chngpage,array,alltt}
\usepackage{amsmath,amsthm,amsfonts,graphicx,amssymb,sidecap,color,subfigure,wrapfig,microtype,cancel,
  bbm}
\usepackage{algorithm}
\usepackage{algpseudocode}

\newtheorem{theorem}{Theorem}
\newtheorem{corollary}{Corollary}[theorem]
\DeclareMathOperator*{\argmax}{argmax}

\newcommand{\binspec}{s}
\newcommand{\hCard}{|H|}
\newcommand{\distp}{\mbox{p}}
\newcommand{\bigo}{\mathcal{O}}
\newcommand{\expect}{\mathbf{E}}
\newcommand{\pep}{x}
\newcommand{\bion}{b_{\pep}}
\newcommand{\yion}{y_{\pep}}
\newcommand{\reals}{\mathbb{R}}
\newcommand{\cR}{\mathbb{R}}
\newcommand{\candidatePeps}{D}
\newcommand{\pepDb}{\mathcal{D}}
\newcommand{\indicator}{\mathbf{1}}
\newcommand{\thomson}{\ensuremath{\mathsf{Th}}}
\newcommand{\obsSpec}{s}
\newcommand{\theoVector}{u}
\newcommand{\procObs}{z}
\newcommand{\cve}[1]{f_{\theta_{#1}}}
\newcommand{\premz}{m^{\obsSpec}}
\newcommand{\preCharge}{c^{\obsSpec}}

\newcommand{\maxth}{\bar{o}}
\newcommand{\pepUni}{\mathbb{P}}

\newcommand{\dideaMaxShift}{L}
\newcommand{\tauCard}{|\tau|}

\newcommand{\dideaShift}{\tau}
\newcommand{\dideaScore}{\psi}

\newboolean{isdraft}
\setboolean{isdraft}{false}
\ifthenelse{\boolean{isdraft}}{
\newcommand{\john}[1]{\textcolor{red}{(John: #1)}}
}{
\newcommand{\john}[1]{}
}

\providecommand{\doshowproof}{false}
\newboolean{isshowproof}
\setboolean{isshowproof}{\doshowproof} 
\ifthenelse{\boolean{isshowproof}}{%
\newcommand{\showproof}[1]{#1}
}{
\newcommand{\showproof}[1]{}
}

\title{Learning Concave Conditional Likelihood Models for Improved
Analysis of Tandem Mass Spectra}

\author{{\bf John T. Halloran} \\
Department of Public Health Sciences \\
University of California, Davis \\
\texttt{jthalloran@ucdavis.edu} \\
\And
{\bf David M. Rocke} \\
Department of Public Health Sciences \\
University of California, Davis \\
\texttt{dmrocke@ucdavis.edu}
}

\hypersetup{bookmarks,colorlinks=true,citecolor=Violet,linkcolor=Mahogany,urlcolor=blue}
\begin{document}

\maketitle

\begin{abstract}
The most widely used technology to identify the proteins present in a
complex biological sample is tandem mass
spectrometry, which quickly produces a large collection of
spectra representative of the \emph{peptides} (i.e., protein subsequences)
present in the original sample.
In this work, we greatly expand the
parameter learning capabilities of a dynamic Bayesian network (DBN)
peptide-scoring algorithm, Didea~\cite{singh2012-didea-uai}, by deriving
emission distributions for which its conditional log-likelihood scoring function
remains concave.  We show that this class of emission
distributions, called \emph{Convex Virtual Emissions} (CVEs),
naturally generalizes the log-sum-exp function while rendering both maximum
likelihood estimation and conditional maximum likelihood estimation
concave for a wide range of Bayesian networks.  Utilizing CVEs in
Didea allows efficient
learning of a large number of parameters while ensuring global
convergence, in stark contrast to Didea's previous parameter learning
framework (which could only learn a single parameter using a costly
grid search) and other trainable models~\cite{halloran2014uai-drip,
  halloran2016dynamic, halloran2017gradients} (which only ensure
convergence to local optima).  The newly trained scoring function
substantially outperforms the state-of-the-art in
both scoring function accuracy and downstream Fisher kernel
analysis. Furthermore, we significantly improve
Didea's runtime performance through successive optimizations to its
message passing schedule and
derive explicit connections between Didea's new concave score and
related MS/MS scoring functions.
\end{abstract}

\section{Introduction}
A fundamental task in medicine and biology is identifying the proteins
present in a complex biological sample, such as a drop of blood.  The
most widely used technology to accomplish this task is \emph{tandem mass
spectrometry} (\emph{MS/MS}), which quickly produces a large collection of
spectra representative of the peptides (i.e., protein subsequences)
present in the original sample.  A critical problem in MS/MS analysis,
then, is the accurate identification of the peptide generating each observed spectrum.  
\begin{figure}[htbp!]
\centering
\includegraphics[trim=0.9in 0.0in 1.0in 0.4in,clip=true,
width=0.8\linewidth]{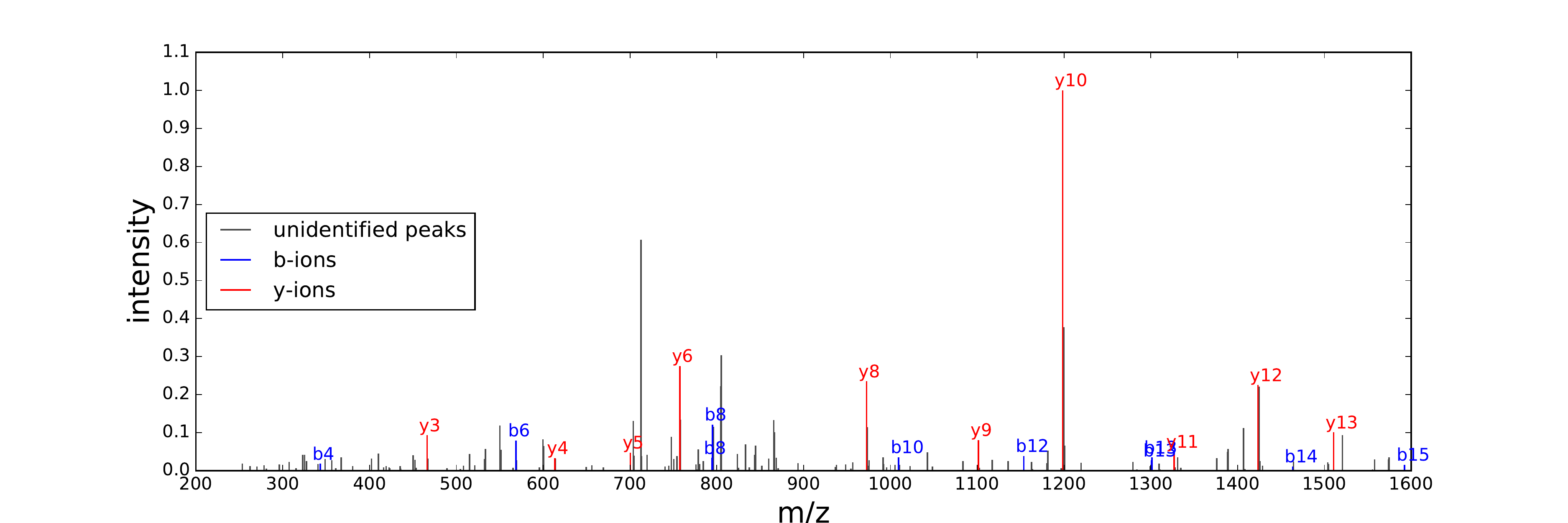}
\caption{{\small Example tandem mass spectrum with precursor charge
  $\preCharge=2$ and generating peptide $\pep =
  \mbox{TGPSPQPESQGSFYQR}$.  Plotted in red and blue are, respectively,
  b- and y-ion peaks (discussed in
  Section~\ref{section:dbsearch}), while unidentified peaks
  are colored gray.}}
\label{fig:exampleSpectrum}
\end{figure}

The most accurate methods which solve this
problem search a database of peptides derived from the mapped
organism of interest.  Such database-search algorithms score peptides
from the database and return the top-ranking peptide per spectrum.
The pair consisting of an observed spectrum and scored peptide
are typically referred to as a \emph{peptide-spectrum match}
(\emph{PSM}).  Many scoring functions have been proposed, ranging from
simple dot-products~\cite{craig:tandem, wenger2013proteomics}, to
cross-correlation based~\cite{eng:approach}, to $p$-value
based~\cite{geer:omssa, kim:msgfPlus,
 howbert:computing}.   Recently, 
dynamic Bayesian networks (DBNs) have been shown to achieve the
state-of-the-art in both PSM identification accuracy and post-search
discriminative analysis, owing to their temporal modeling
capabilities, parameter learning capabilities, and generative nature.



The first such DBN-based scoring function,
Didea~\cite{singh2012-didea-uai}, used sum-product inference to
efficiently compute a log-posterior score for highly accurate PSM
identification.  
However, Didea utilized a
complicated emission distribution for which only a single
parameter could be learned through a costly grid search.
Subsequently, a DBN for Rapid Identification of
Peptides (DRIP)~\cite{halloran2016dynamic, halloran2018analyzing},
was shown in~\cite{halloran2014uai-drip} to outperform Didea due to
its ability to generatively learn a large number of model parameters.
Most recently, DRIP's generative nature was further exploited to
derive log-likelihood gradients detailing the manner in which peptides
align with observed spectra~\cite{halloran2017gradients}.  Combining
these gradients with a discriminative
postprocessor~\cite{kall:semi-supervised}, the resulting DRIP Fisher
kernel substantially improved upon all state-of-the-art methods for
downstream analysis on a large number of datasets.  

However, while DRIP significantly improves several facets of MS/MS
analysis due to its parameter learning capabilities, these
improvements come at high runtime cost.  In practice, DRIP
inference is slow due to its large model complexity
(the state-space grows exponentially in the lengths of both the
observed spectrum and peptide).  For instance, DRIP
search required an order of magnitude longer than the slowest
implementation of Didea for the timing tests in
Section~\ref{section:fasterResults}.  Herein, we greatly improve upon
all the analysis strengths provided by DRIP using the much faster
Didea model.  Furthermore, we optimize Didea's message passing
schedule for a $64.2\%$ speed improvement, leading to runtimes two
orders of magnitude faster than DRIP and comparable to less
accurate(but widely used) methods.  Thus, the work described herein
not only improves upon state-of-the-art DBN analysis for effective
parameter learning, scoring function accuracy, and downstream Fisher
kernel recalibration, but also renders such analysis practical by
significantly decreasing state-of-the-art DBN inference time.

In this work, we begin by discussing relevant MS/MS background and
previous work.  We then greatly expand the parameter learning
capabilities of Didea by deriving a class of Bayesian network (BN)
emission distributions for which both maximum likelihood learning and,
most importantly, conditional maximum likelihood learning are
concave.  Called Convex Virtual Emissions (CVEs), we show that this
class of emission
distributions generalizes the widely used log-sum-exp function and
naturally arises from the solution of a nonlinear differential
equation representing convex conditions for general BN emissions.  We
incorporate CVEs into Didea to quickly and
efficiently learn a substantial number of model parameters,
considerably improving upon the previous learning
framework.  The newly trained model drastically improves PSM
identification accuracy, outperforming all state-of-the-art methods
over the presented datasets; at a strict FDR of $1\%$ and averaged
over the presented datasets, the trained
Didea scoring function identifies $16\%$ more spectra than DRIP and
$17.4\%$ more spectra than the highly accurate and widely used
MS-GF+~\cite{kim:msgfPlus}.  Under the newly parameterized model, we
then derive a bound explicitly relating Didea's score to the popular
XCorr scoring function, thus providing potential avenues to train
XCorr using the presented parameter learning work.

With efficient parameter learning in place, we next utilize
the new Didea model to improve MS/MS recalibration performance.  We
use gradient information derived from Didea's conditional
log-likelihood in the feature-space of a kernel-based
classifier~\cite{kall:semi-supervised}.  Training the resulting
conditional Fisher kernel substantially improves upon the
state-of-the-art recalibration performance previously achieved by
DRIP; at a strict FDR of $1\%$, discriminative recalibration using
Didea's conditional Fisher kernel results in an average $11.3\%$ more
identifications than using the DRIP Fisher kernel.  Finally, we
conclude with a discussion of several avenues for future work.

\section{Tandem mass spectrometry}\label{section:msms}
With a complex sample as input, a typical MS/MS experiment begins by cleaving the
proteins of the sample into peptides using a digesting enzyme, such as
trypsin.  The digested peptides are then separated via liquid
chromatography and undergo two rounds of mass spectrometry.  The first
round of mass spectrometry measures the mass and charge of the intact
peptide, referred to as the \emph{precursor mass} and \emph{precursor
  charge}, respectively.  Peptides are then fragmented into prefix and
suffix ions.  The mass-to-charge (\emph{m/z}) ratios of the resulting
fragment ions are measured in the second round of mass spectrometry,
producing an observed spectrum of m/z versus intensity values
representative of the fragmented peptide.  The output of this overall
process is a large collection of spectra (often numbering in the
hundreds-of-thousands), each of which is representative of a peptide
from the original complex sample and requires identification.  The
x-axis of such observed spectra denotes m/z, measured in thomsons
(\thomson), and y-axis measures the intensity at a particular m/z
value.  A
sample such observed spectrum is illustrated in
Figure~\ref{fig:exampleSpectrum}.

\subsection{Database search and theoretical
  spectra}\label{section:dbsearch}
Let $s \in S$ be an observed spectrum with precursor m/z $\premz$ and
precursor charge $\preCharge$, where $S$ is the universe of tandem
mass spectra.  The generating peptide of $s$ is
identified by searching a database of peptides, as follows.
Let $\pepUni$ be the universe of all peptides and $x \in \pepUni$ be an
arbitrary peptide of length $l$.  $x = x_1\dots x_l$ is a
string comprised of characters called \emph{amino acids}, the
dictionary size of which are 20.  We denote peptide substrings as
$x_{i:j} = x_i,\dots,x_j$, where $i > 0, j \leq l, i<j$, and the mass of
$x$ as $m(x)$.  Given a
peptide database $\pepDb \subseteq \pepUni$, the set of peptides
considered is constrained to those within a precursor
mass tolerance window $w$ of $\premz$.  The set of \emph{candidate
peptides} to be scored is thus
$\candidatePeps (s, \pepDb, w) = \{ x: x \in \pepDb,
|\frac{m(x)}{\preCharge} - \premz| \leq w \}$.  Using a scoring
function $\psi: \pepUni \times S \to \reals$, a database search
outputs the top-scoring PSM,  $x^* = \argmax_{x \in \pepDb}\psi(x,s)$.


In order to score a PSM, the idealized fragment ions of $x$
are first collected into a theoretical
spectrum.  The most commonly encountered fragment ions are called
\emph{b-ions} and \emph{y-ions}.  
B- and y-ions
correspond to prefix and suffix mass pairs, respectively, such that
the precursor charge $c^s$ is divided amongst the pair.  For b-ion
charge $c_b \leq \preCharge$, the $k$th b-ion and the accompanying
y-ion are then, respectively,
{\small
\begin{align*}
b(m(x_{1:k}),c_b) =& \frac{m(x_{1:k}) + c_b}{c_b} = \frac{\left [\sum_{i =
    1}^k m(x_i)\right] + c_b}{c_b}, \;\;\; y(m(x_{k+1:l}),c_y) =\frac{\left[\sum_{i = k+1}^l m(x_i)\right] + 18+ c_y}{c_y},
\end{align*}}
where $c_y$ is the y-ion charge, the b-ion offset corresponds to a
$c_b$ charged hydrogen atom, and the y-ion offset corresponds to a
$c_y$ charged hydrogen atom plus a water molecule. For singly charged
spectra $\preCharge=1$, only singly charged fragment ions are
detectable, so that $c_b = c_y = 1$.  For higher precursor
charge states $\preCharge \geq 2$, the total charge is split between
each b- and y-ion pair, so that $0 < c_b < \preCharge$ and
$c_y = \preCharge - c_b$.  The annotated b- and y-ions of an
identified observed spectrum are illustrated in
Figure~\ref{fig:exampleSpectrum}.




\section{Previous work}\label{section:previousWork}
Many database search scoring algorithms have been proposed, each of
which is characterized by the scoring function they employ.  These
scoring functions have ranged from
dot-products (X!Tandem~\cite{craig:tandem} and
Morpheus~\cite{wenger2013proteomics}), to cross-correlation based
(XCorr~\cite{eng:approach}), to exact $p$-values computed
over linear scores~\cite{kim:msgfPlus, howbert:computing}.  Recently,
DBNs have been used to substantially improve upon the accuracy of
previous approaches.

In the first such DBN, Didea~\cite{singh2012-didea-uai}, 
the time series being modeled is the sequence of a peptide's
amino acids (i.e., an amino acid is observed in each frame)
and the quantized observed spectrum is observed in each frame.  In
successive frames, the sequence of b- and y-ions are computed and used
as indices into the observed spectrum via virtual
evidence~\cite{pearl:probabilistic}.    A hidden
variable in the first frame, corresponding to the amount to
shift the observed spectrum by, is then marginalized in order to compute a conditional
log-likehood probability consisting of a foreground score minus a
background score, similar in form to XCorr
(described in Section~\ref{section:xcorr}).  The resulting scoring function
outperformed the most accurate scoring algorithms at the time
(including MS-GF+, then called MS-GFDB) on a majority of datasets.
However, parameter learning in the model was severely limited and
inefficient; a single hyperparameter controlling the reweighting of
peak intensities was learned via an expensive grid search, requiring
repeated database searches over a dataset.

Subsequent work saw the introduction of
DRIP~\cite{halloran2014uai-drip, halloran2016dynamic}, a DBN
with substantial parameter
learning capabilities.  In DRIP, the time
series being modeled is the sequence of observed
spectrum peaks (i.e., each frame in DRIP corresponds to an observed
peak) and two types of prevalent phenomena are explicitly modeled via
sequences of random variables: spurious observed peaks (called
insertions) and absent theoretical peaks (called deletions).  A
large collection of Gaussians parameterizing the m/z axis are
generatively learned,
via expectation-maximization (EM)~\cite{dempster:maximum}, and used to
score observed peaks.  DRIP then uses max-product inference to
calculate the most probable sequences of insertions and deletions in
order to score PSMs.

In practice, the majority of PSM scoring functions discussed are
typically poorly calibrated, i.e., it is often difficult to compare
the PSM scores across different spectra.  In order to combat such poor
calibration, postprocessors are commonly employed to recalibrate PSM
scores~\cite{kall:semi-supervised, spivak:improvements,
  spivak:direct}.  In recent work, DRIP's generative framework was
further exploited to calculate highly accurate features based on the
log-likelihood gradients of its learnable parameters.  Combining these
new gradient-based features with a popular kernel-based classifier for
recalibrating PSM scores~\cite{kall:semi-supervised}, the resulting
Fisher kernel was shown to significantly improve postprocessing
accuracy~\cite{halloran2017gradients}.

\section{Didea}\label{section:didea}
\begin{figure}[htbp!]
\begin{center}
\includegraphics[page=9,trim=0.2in 3.4in 0.45in 0.0in,clip=true,
width=0.75\linewidth]{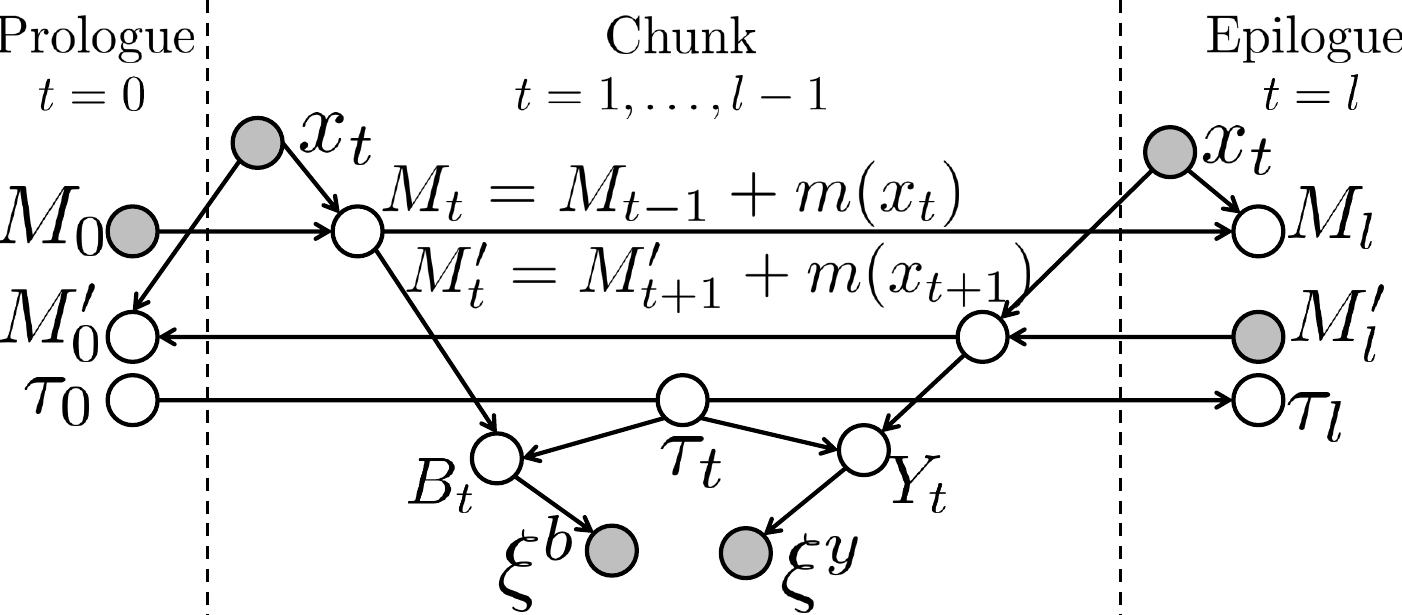}
\end{center}
\caption{{\small Graph of Didea.  Unshaded nodes are hidden, shaded nodes are
  observed, and edges denote deterministic functions of parent variables.}}
\label{fig:dideaGraph}
\end{figure}
We now derive Didea's scoring function in detail.  The graph of Didea is displayed in
Figure~\ref{fig:dideaGraph}.  Shaded variables are observed and
unshaded variables are hidden (random).  Groups of variables are
collected into time instances called \emph{frames}, where the first
frame is called the prologue, the final frame is called the epilogue,
and the \emph{chunk} dynamically expands to fill all frames
in between.
  Let $0 \leq t
\leq l$ be an arbitrary frame.  The amino acids
of a peptide are observed in each frame after the prologue.  The variable
$M_t$ successively accumulates the prefix masses of the peptide such
that $p(M_0 = 0) = 1$ and
$p(M_t = M_{t-1}+m(x_t) | M_{t-1}, x_t) = 1$, while
the variable $M_t'$ successively accumulates the suffix masses of the
peptide such that $p(M_l' = 0) = 1$ and $p(M'_t = M'_{t+1}+m(x_{t+1}) |
M'_{t+1}, x_{t+1}) = 1$.  Denoting the maximum spectra shift as
$\dideaMaxShift$, the shift variable $\tau_0\in
[-\dideaMaxShift,\dideaMaxShift]$ is hidden, uniform, and
deterministically copied by its descendents in successive frames, such
that $p(\tau_t =\bar{\tau}| \tau_0= \bar{\tau}) = 1$ for $t>1$.

Let $\binspec \in \reals^{\maxth + 1}$ be the binned observed
spectrum, i.e., a vector of length $\maxth + 1$ whose $i$th element is
$s(i)$, where $\maxth$ is the maximum observable
discretized m/z value.  Shifted versions of the
$t$th b- and y-ion pair (where the shift is denoted by subscript) are
deterministic functions of the shift variable as well as prefix and
suffix masses, i.e., 
$ p(B_t = b_{\tau_t}(M_t,1) | M_t, \tau_t) = p(B_t =
\max(\min(b(M_t,1) - \tau_t, 0), \maxth) | M_t, \tau_t), p(Y_t =
y_{\tau_t}(m(x_{t+1:l}), 1) | M'_{t}, \tau_t) =  p(Y_t =
\max(\min(y(m(x_{t+1:l}), 1) - \tau,0),\maxth) | M'_{t}, \tau_t)$,
respectively.  $\xi^b$ and $\xi^y$ are \emph{virtual evidence
  children}~\cite{pearl:probabilistic}, i.e., leaf nodes whose conditional distribution need not
be normalized (only non-negative) to compute posterior probabilities
in the DBN.  A comprehensive overview of virtual evidence is available
in~\cite{halloranThesis2016}.  $\xi^b$ and $\xi^y$ compare the b- and
y-ions, respectively, to the observed spectrum, such that $p(\xi^b |
B_t) = f(\binspec (b_{\tau_t}(m(x_{1:t}), 1))), p(\xi^y | Y_t) =
f(\binspec (y_{\tau_t}(m(x_{t+1:l}), 1)))$, where $f$ is a non-negative emission
function.

Let $\indicator_{\{\cdot\}}$ denote the indicator function.  Didea's
log-likelihood is then
$\log \distp ( \dideaShift_0=\bar{\tau},\pep,
s)$
{\small
\begin{align*}
=&\log{p(\tau_0 = \bar{\tau})p(M_0)p(M'_0| M'_1,
x_1)p(M'_l)p(M_l| M_{l-1},x_1)} +\\
&\,\log\prod_{t=1}^l[p(\tau_t |
\tau_{t-1})p(M_t| M_{t-1},x_t)p(M'_t| M'_{t+1},
x_{t+1})p(B_t | M_t, \tau_t)p(Y_t|M'_t,
\tau_t)p(\xi^b | B_t)p(\xi^y | Y_t)]\\
=&\log p(\tau_0 = \bar{\tau})+
\log \prod_{t=1}^{l-1} \bigl(\indicator_{\{ \tau_t = \bar{\tau} \wedge M_t
  = m(x_{1:t}) \wedge M'_t = m(x_{t+1:l}) \}}p(B_t | M_t, \tau_t)p(Y_t|M'_t,
\tau_t)p(\xi^b | B_t)p(\xi^y | Y_t) \bigr)\\
=&\log p(\tau_0 = \bar{\tau})+
\log \prod_{t=1}^{l-1}p(\xi^b | b_{\bar{\tau}}(M_t,
  1))p(\xi^y |y_{\bar{\tau}}(M'_t, 1))]\\
=&\log p(\tau_0 = \bar{\tau}) + \sum_{t=1}^{l-1} \bigl (\log
f(\binspec _{\bar{\tau}}(b(m(x_{1:t}),
1)))+ \log f(\binspec _{\bar{\tau}}(y(m(x_{t+1:l}),1))) \bigr).
\end{align*}
}

In order to score PSMs, Didea computes the conditional
log-likelihood
{\small
\begin{align}
\dideaScore(\obsSpec,
\pep)&= \log{\distp
  (\dideaShift_0=0 | \pep, \binspec})=\log{\distp ( \dideaShift_0=0,\pep, \binspec)} -
\log{\sum_{\bar{\tau}=-\dideaMaxShift}^\dideaMaxShift\distp(\dideaShift_0=\bar{\tau}) \distp(\pep,
  \binspec_{\bar{\tau}}|\dideaShift_0=\bar{\tau})}\nonumber\\
& =\log{\distp ( \dideaShift_0=0,\pep, \binspec)} -
\log{\frac{1}{|\dideaShift_0|}\sum_{\bar{\tau}=-\dideaMaxShift}^\dideaMaxShift\distp(\pep,
  \binspec_{\dideaShift}|\dideaShift_0=\bar{\tau})}.\label{eq:dideaScoringFunction}
\end{align}}
As previously mentioned, $\dideaScore(\obsSpec,
\pep)$ is a foreground score minus a background score, where the
background score consists of averaging over $|\tau_0|$
shifted versions of the foreground score, much like the XCorr scoring
function.  Thus, Didea may be thought of as a probabilistic analogue
of XCorr.

\subsection{Convex Virtual Emissions for Bayesian
  networks}\label{section:cves}
Consider an arbitrary Bayesian network where the observed variables are leaf
nodes, as is common in a large number of time-series models such as hidden Markov
models (HMMs), hierarchical
HMMs~\cite{murphy2002linear},  DBNs for
speech recognition~\cite{deng2006dynamic}, hybrid
HMMs/DBNs~\cite{dahl2012context}, as well as DRIP and Didea.  Let $E$
be the set of observed random variables, $H$ be the hypothesis space
composed of the cross-product of the $n$ hidden discrete random
variables in the network, and $h \in H$ be an arbitrary
\emph{hypothesis} (i.e., an instantiation of the hidden variables).
As is the case in Didea, often desired is the log-posterior
probability $\log p(h | E) = \log \frac{p(h,
  E)}{p(E)} = \log \frac{p(h,
  E)}{\sum_{\bar{h} \in H} p(\bar{h},E)} = \log p(h, E) - \log \sum_{h \in H}
  p(h)p(E | h).$
Under general assumptions, we'd like to find emission functions for
which $\log p(h | E)$ is concave.

Assume $p(h)$ and $p(E | h)$ are non-negative, that the emission
density $p(E | h)$ is parameterized
by $\theta$ (which we'd like to learn), and that there is a parameter
$\theta_h$ to be learned for every hypothesis of latent variables
(though if we have fewer parameters, parameter estimation becomes
strictly easier).  We make this parameterization explicit by denoting
the emission distributions of interest as $p_{\theta_h}(E | h)$.
Assume that $p_{\theta_h}(E | h)$ is smooth on $\cR$ for all $h \in
H$.  Applying virtual evidence for such models, $p_{\theta_h}(E | h)$
need not be normalized for posterior inference (as well as Viterbi inference and
comparative inference between sets of observations).

Given the factorization of the joint
distribution described by the BN, the quantity $p_{\theta_h}(h,E) = p(h)p_{\theta_h}(E|h)$
may often be efficiently computed for any given $h$.  Thus, the
computationally difficult portion of
$\log p(h | E)$ is the calculation of the log-likelihood in
the denominator, wherein all hidden
variables are marginalized over.  We therefore first seek
emission functions for which the log-likelihood $\log p(E) = \log \sum_{h \in H}
  p(h)p_{\theta_h}(E | h)$ is convex.  For such emission functions, we
  have the following theorem.
\begin{theorem}\label{theorem:generalizedLse}
The unique convex functions of the form $\log \sum_{h \in H}
p(h)p_{\theta_h}(E|h)$, such that $(p'_{\theta_h}(E|h))^2 -
p''_{\theta_h}(E|h)p_{\theta_h}(E|h) = 0$,  are $\log \sum_{h \in H}
p(h)p_{\theta_h}(E|h) = \log \sum_{h \in H}\alpha_h
e^{\beta_h\theta_h}$, where $\alpha_h = p(h)a_h$ and $a_h, \beta_h$
are hyperparameters.
\end{theorem}
The full proof of Theorem~\ref{theorem:generalizedLse} is given
in Appendix~\ref{appendix:cveProof}.  The nonlinear differential equation
$(p'_{\theta_h}(E|h))^2 - p''_{\theta_h}(E|h)p_{\theta_h}(E|h) = 0$
describes the curvature of the desired emission functions and arises
from the necessary and sufficient conditions for twice differentiable
convex functions (i.e., the Hessian must be p.s.d.) and the
Cauchy-Schwarz inequality.  Particular values of the
hyperparameters $a_h$ and $\beta_h$ correspond to unique initial
conditions for this nonlinear differential equation.
Note that when $\alpha_h = 1, p_{\theta_h}(E | h) = e^{\theta_h}$, we
have the well-known log-sum-exp (LSE) convex function.  Thus, this
result generalizes the LSE function to a broader class of convex
functions.

We call the unique class of convex functions which arise from solving
the nonlinear differential in Theorem~\ref{theorem:generalizedLse},
$p_{\theta_h}(E|h) = a_h e^{\beta_h\theta_h}$, Convex Virtual
Emissions (CVEs).  Note that utilizing CVEs, maximimum likelihood
estimation (i.e., $\argmax_{\theta} -\log \sum_{h \in H}
p(h)p_{\theta_h}(E|h)$) is thus concave and guaranteed to converge to a
global optimum.  
Furthermore, and most importantly for Didea, we have
the following result for the conditional log-likelihood (the full
proof of which is in Appendix~\ref{appendix:cveProof}).
\begin{corollary}\label{corollary1}
For convex $\log p(E) = \log \sum_{h \in H} p(h)p_{\theta_h}(E | h)$ such that
$(p'_{\theta_h}(E|h))^2 - p''_{\theta_h}(E|h)p_{\theta_h}(E|h) =
0$, the log-posterior $\log p_{\theta}(h | E)$ is concave in $\theta$.
\end{corollary}
Thus, utilizing CVEs, conditional maximum likelihood estimation is
also rendered concave.

\subsection{CVEs in Didea}\label{section:veLse}
In~\cite{singh2012-didea-uai}, the virtual evidence emission function
to score peak intensities was $f_{\lambda}(\binspec (i)) = 1-\lambda
e^{-\lambda} + \lambda e^{-\lambda(1-\binspec(i))}$.  Under this function,
Didea was shown to perform well on a variety of datasets.  However,
this function is non-convex and does not permit efficient parameter
learning; although only a single model parameter, $\lambda$, was
trained, learning required a grid search wherein each step consisted
of a database search over a dataset and
subsequent target-decoy analysis to assess each new parameter value.
While this training scheme is already costly and impractical, it
quickly becomes infeasible when looking to learn more than a single
model parameter.

We use CVEs to
render Didea's conditional log-likelihood concave given a large number of
parameters.  To efficiently learn a distinct
observation weight $\theta_{\tau}$ for each spectral shift $\tau \in
[-L, L]$, we thus utilize the emission function $f_{\theta_\tau}(\binspec (i)) =
e^{\theta_{\tau} \binspec (i)}$. 
Denote the set of parameters per spectra shift as $\theta =
\{\theta_{-L}, \dots, \theta_L\}$.
Due to the concavity of
Equation~\ref{eq:dideaScoringFunction} using $f_{\theta_\tau}$ under
Corollary~\ref{corollary1}, given $n$ PSM training pairs $(s^1, x^1),
\dots, (s^n, x^n)$, the learned parameters $\theta^* =
\argmax_{\theta} \sum_{i = 1}^n\psi_{\theta}(s^i, x^i)$ are
guaranteed to converge to a global optimum.
Further analysis of Didea's scoring function under
this new emission function may be found
in Appendix~\ref{section:veLse}, including the derivation of the
new model's gradients (i.e., conditional Fisher scores).

\subsubsection{Relating Didea's conditional log-likelihood to XCorr
  using CVEs}\label{section:xcorr}
XCorr~\cite{eng:approach}, the very first database search scoring
function for peptide identification, remains one of the most widely
used tools in the field today.  Owing to its prominence, XCorr remains
an active subject of analysis and continuous
development~\cite{klammer:statistical, park:rapid, eng:fast,
  Diament2011, howbert:computing, mcilwain:crux, eng:comet,
  halloran2017gradients}.  As previously noted, the scoring functions
of XCorr and Didea share several similarities in form, where, in fact,
the former served as the motivating example
in~\cite{singh2012-didea-uai} for both the design of the
Didea model and its posterior-based scoring function.
While cosmetic similarities have thus far been noted, the
reparameterization of Didea's conditional log-likelihood using CVEs
permits the derivation of an explicit relationship between the two.

Let $u$ be the
theoretical spectrum of peptide $x$.  As with Didea, let $L$ be the
maximum spectra shift considered and,
for shift $\tau$, denote a vector shift via subscript, such that
$s_{\tau}$ is the vector of observed spectrum elements shifted by
$\tau$ units.  In order to compare $u$ and $\binspec$, XCorr is thus
computed as $\mbox{XCorr}(\binspec, \pep) = \theoVector^T\binspec-
\frac{1}{2L+1}\sum_{\tau=-L}^{L}\theoVector^T\binspec_{\tau}.$
Intuitively, the cross-correlation background term is meant to
penalize overfitting of the theoretical spectrum.  Under the newly
parameterized Didea conditional log-likelihood described herein, we
have the following theorem explicitly relating the XCorr and Didea
scoring functions. 

\begin{theorem}\label{theorem:xcorrDidea}
Assume the PSM scoring function $\dideaScore(s,x)$ is that of Didea
(i.e., Equation~\ref{eq:dideaScoringFunction}) where the emission
function $f_{\theta_{\tau}}(s(i))$ has uniform weights $\theta_{i} =
\theta_{j}$, for $i,j \in [-L, L]$.  Then $\dideaScore(s,x) \leq
\bigo(\mbox{XCorr}(\obsSpec, \pep))$.
\end{theorem}
The full proof of Theorem~\ref{theorem:xcorrDidea} may be found
in Appendix~\ref{section:xcorrbound}.
Thus, Didea's scoring
function effectively serves to lower bound XCorr.
This opens possible avenues for extending the learning results
detailed herein to the widely used XCorr function.  For instance, a
natural extension is to use a variational Bayesian inference approach and 
learn XCorr parameters through iterative maximization of the Didea
lower bound, made efficient by the concavity of new
Didea model derived in this work.

\subsection{Faster Didea sum-product inference}\label{section:faster}
We successively improved Didea's inference time
when conducting a database search using the intensive charge-varying
model (discussed in Appendix~\ref{section:chargeIntegration}).  Firstly, we
removed the need for a backward pass by keeping
track of the foreground log-likelihood during the forward pass (which
computes the background score, i.e., the probability of evidence in
the model).  Next, by exploiting the symmetry of the spectral shifts,
we cut the effective cardinality of $\tau$ in half during
inference. While this requires
twice as much memory in practice, this is not close to being prohibitive
on modern machines.  Finally, a large portion of the speedup was
achieved by offsetting the virtual evidence vector by $|\tau|$ and
pre/post buffering with zeros and offsetting each computed b- and y-ion by
$|\tau|$.  Under this construction, the scores do not change, but,
during inference, we are able to shift each computed b- and y-ion by
$\pm \tau$ without requiring any bound checking.  Hashing virtual
evidence bin values by b-/y-ion value and $\tau$ was also
pursued, but did not offer any runtime benefit over the
aforementioned speedups (due to the cost of constructing the hash
table per spectrum).

\section{Results}\label{section:results}
In practice, assessing peptide identification accuracy is made
difficult by the lack of ground-truth encountered in real-world data.
Thus, it is most common to estimate the \emph{false discovery rate}
(\emph{FDR})~\cite{benjamini:controlling} by searching a decoy
database of peptides which
are unlikely to occur in nature, typically generated by
shuffling entries in the
target database~\cite{keich2015improved}.  For a particular
score threshold, $t$, the FDR is calculated as the proportion of
decoys scoring better than $t$ to the number of targets scoring better
than $t$.  Once the target and decoy PSM scores are calculated, a
curve displaying the FDR threshold versus the number of correctly
identified targets at each given threshold may be calculated.  In
place of FDR along the x-axis, we use the
\emph{q-value}~\cite{keich2015improved},  defined to be the minimum
FDR threshold at which a
given score is deemed to be significant.  As many applications
require a search algorithm perform well at low thresholds, we only
plot $q \in [0, 0.1]$.  

\begin{figure*}[htbp!]
  \centering
  \subfigure{\raisebox{7.0mm}{\includegraphics[trim=3.8in 1.03in 0.3in 2.0in,
    clip=true,scale=0.68]{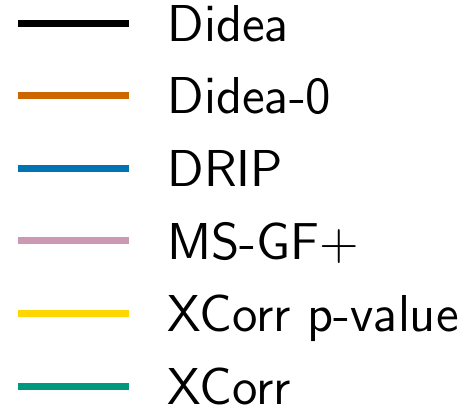}}}
  \subfigure[Worm-1]{\includegraphics[trim=0.0in 0.0in 0.0in 0.05in,
    clip=true,scale=0.30]{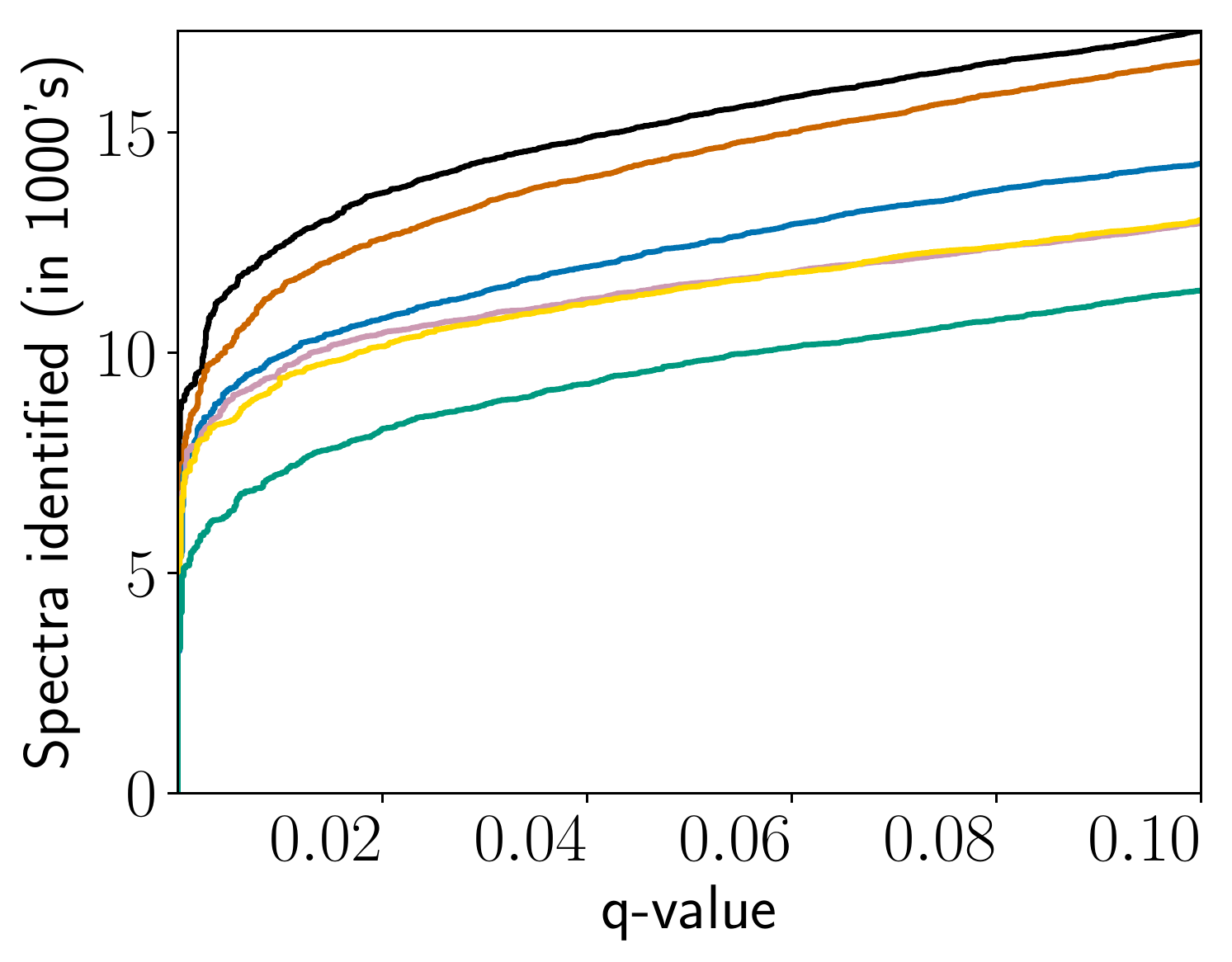}}
  \subfigure[Worm-2]{\includegraphics[trim=0.45in 0.0in 0.0in 0.05in,
    clip=true,scale=0.30]{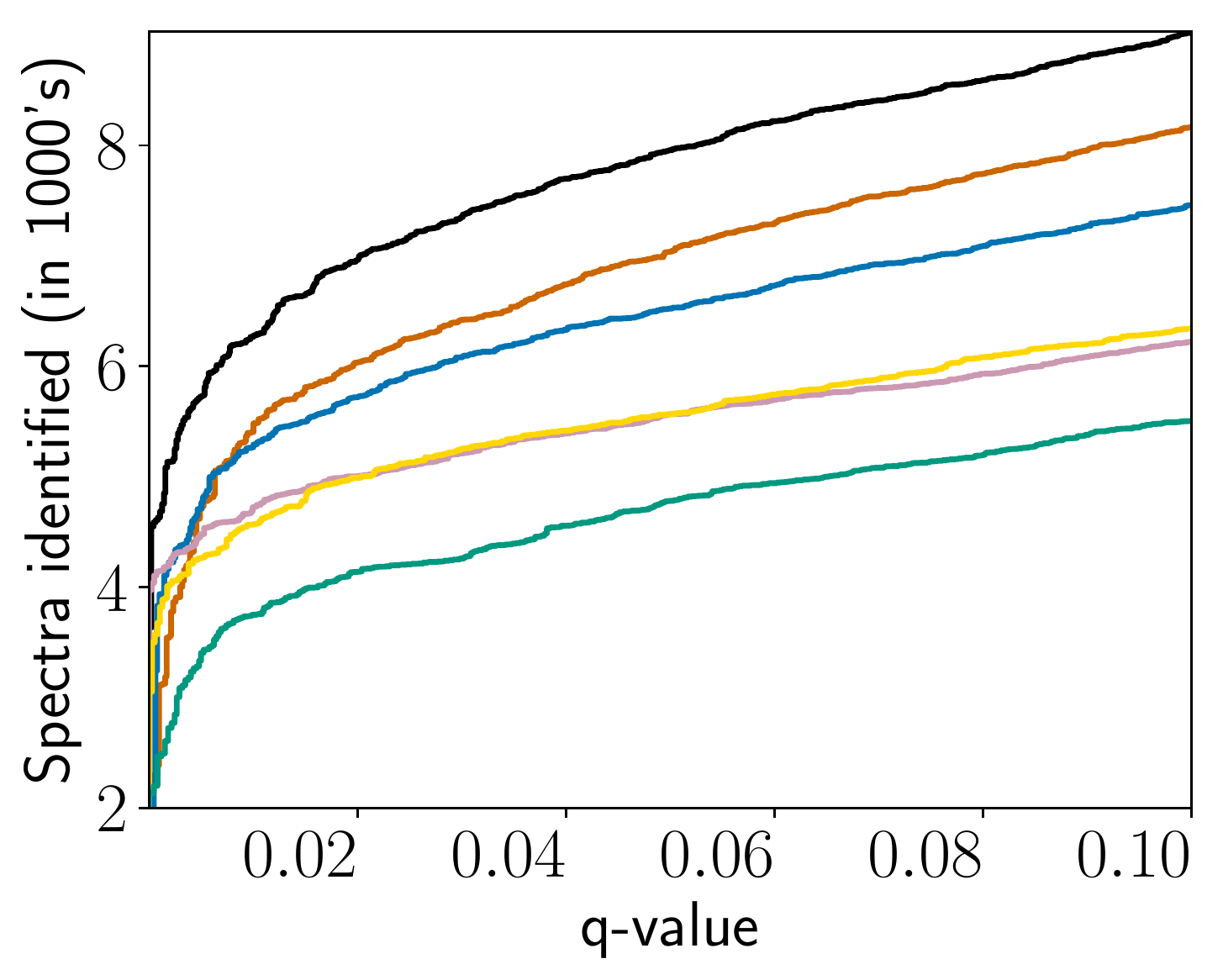}}
  \subfigure[Worm-3]{\includegraphics[trim=0.0in 0.0in 0.0in 0.05in,
    clip=true,scale=0.30]{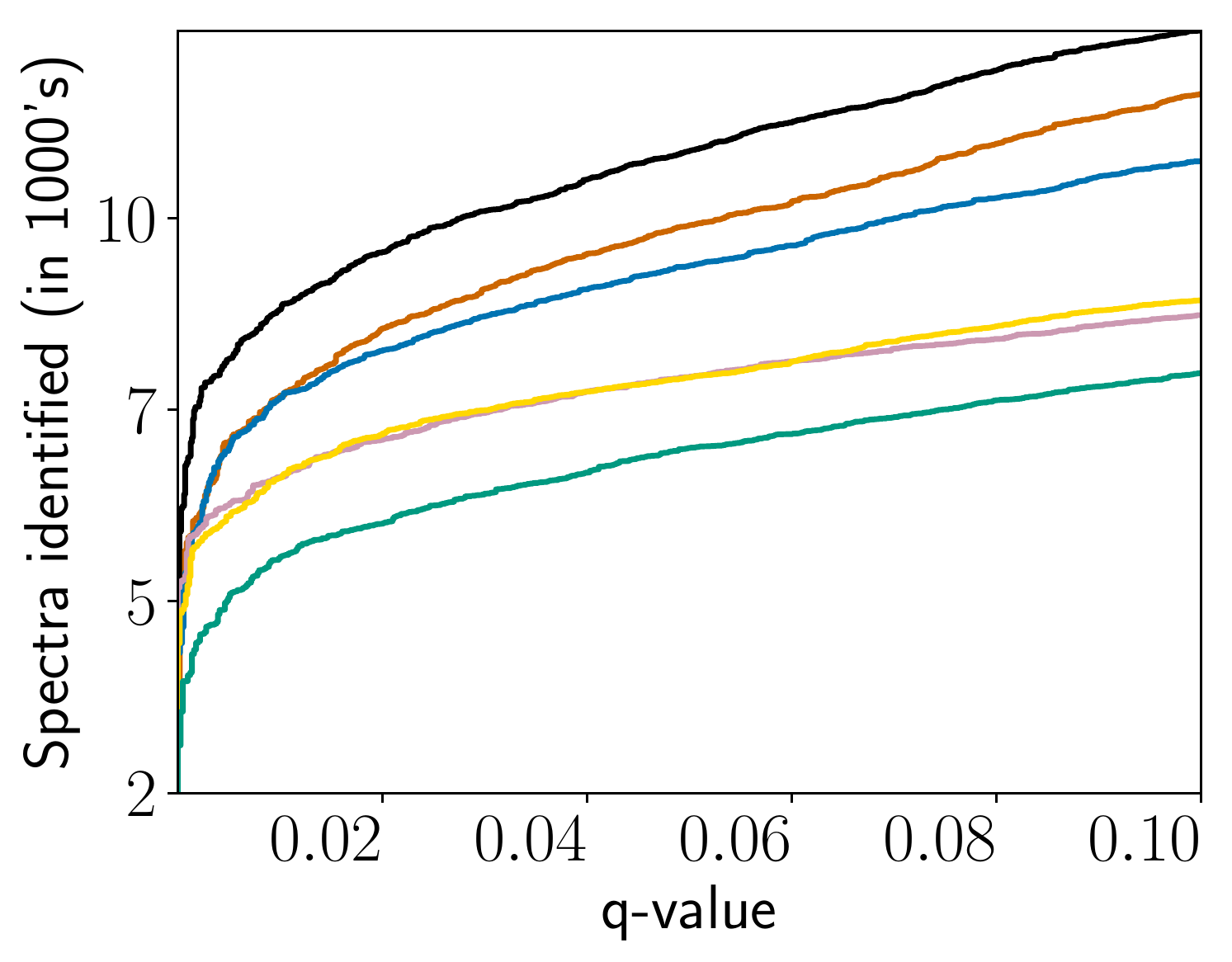}}
  \subfigure[Worm-4]{\includegraphics[trim=0.45in 0.0in 0.0in 0.05in,
    clip=true,scale=0.30]{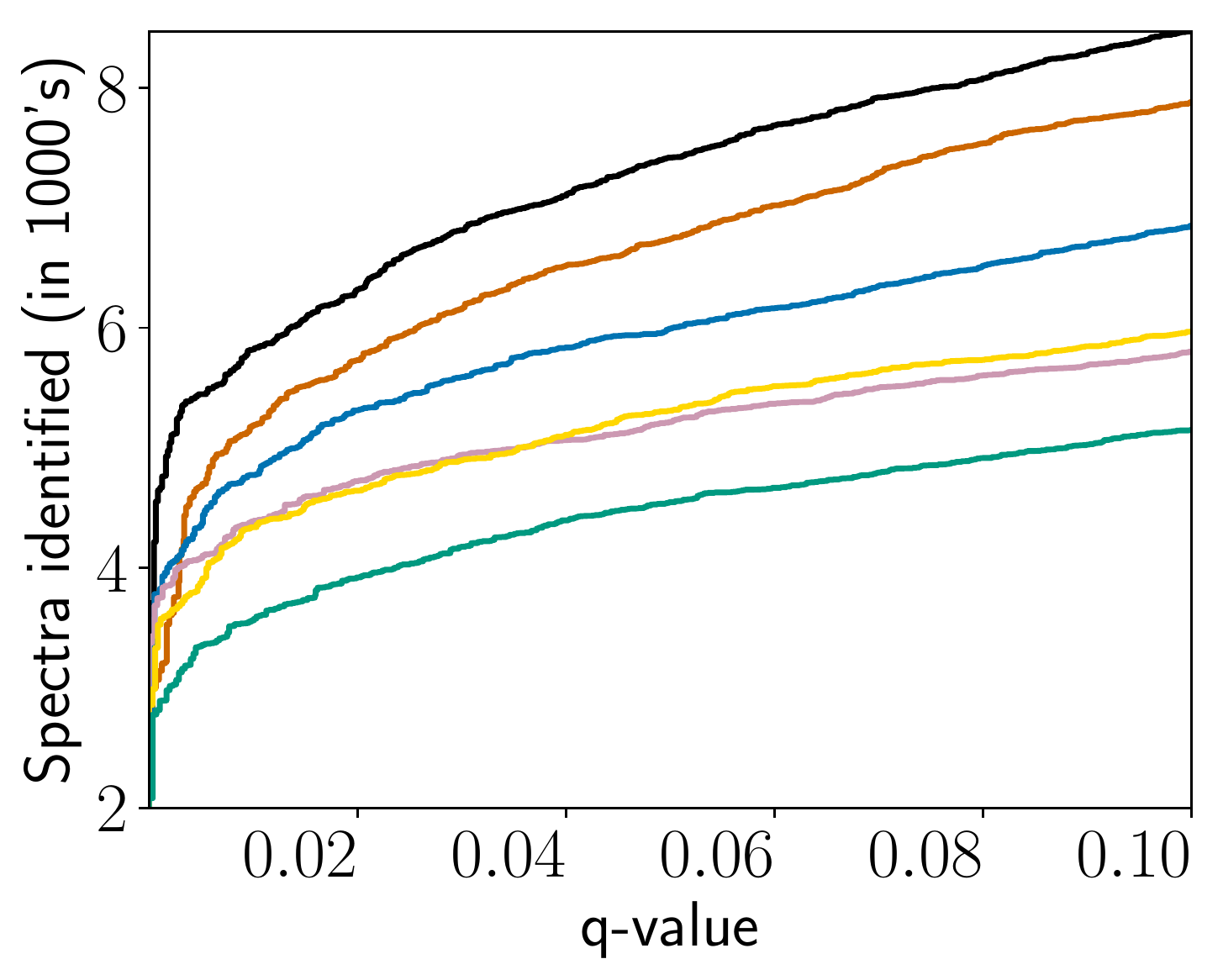}}
  \subfigure[Yeast-1]{\includegraphics[trim=0.45in 0.0in 0.0in 0.05in,
    clip=true,scale=0.30]{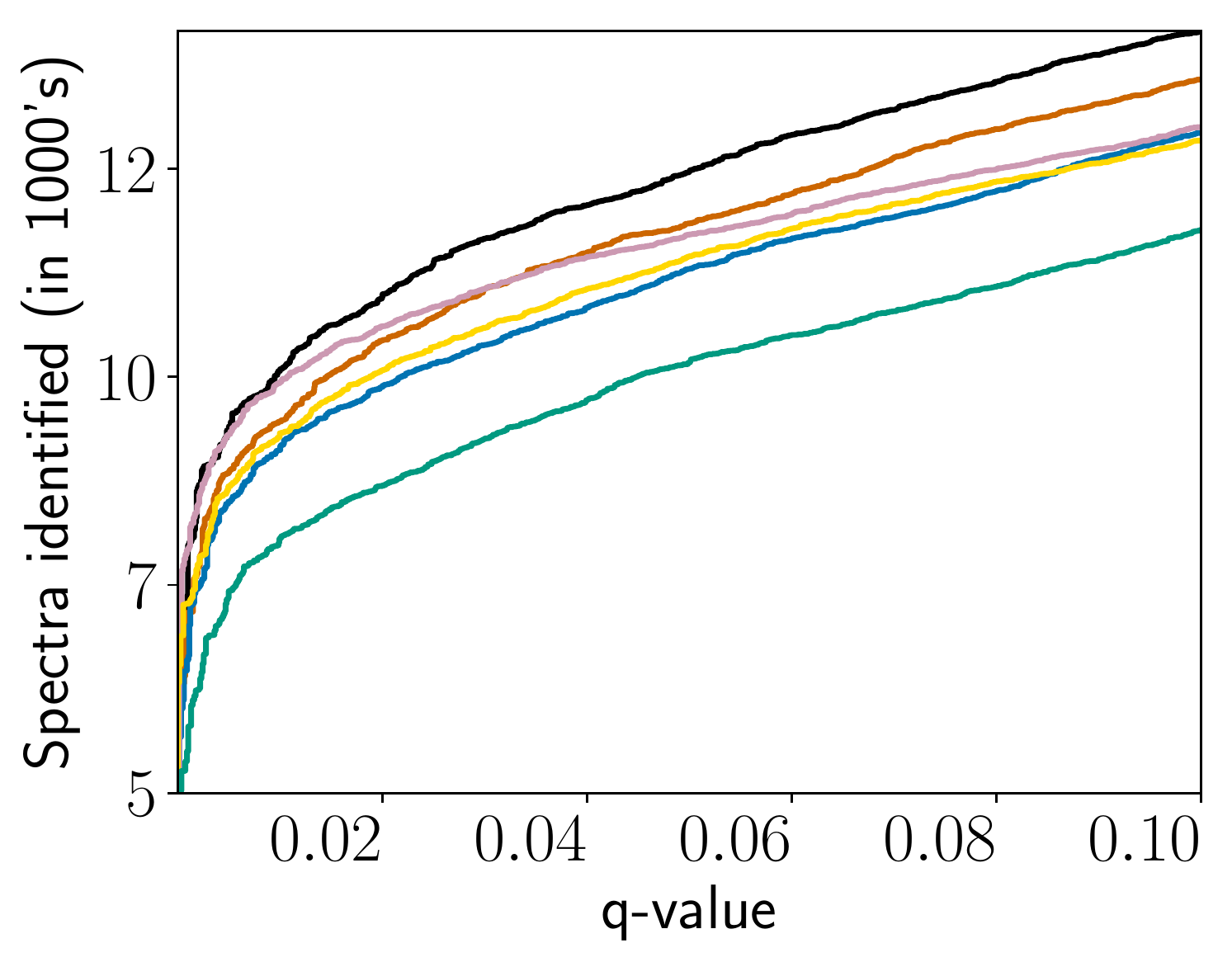}}
  \subfigure[Yeast-2]{\includegraphics[trim=0.0in 0.0in 0.0in 0.05in,
    clip=true,scale=0.30]{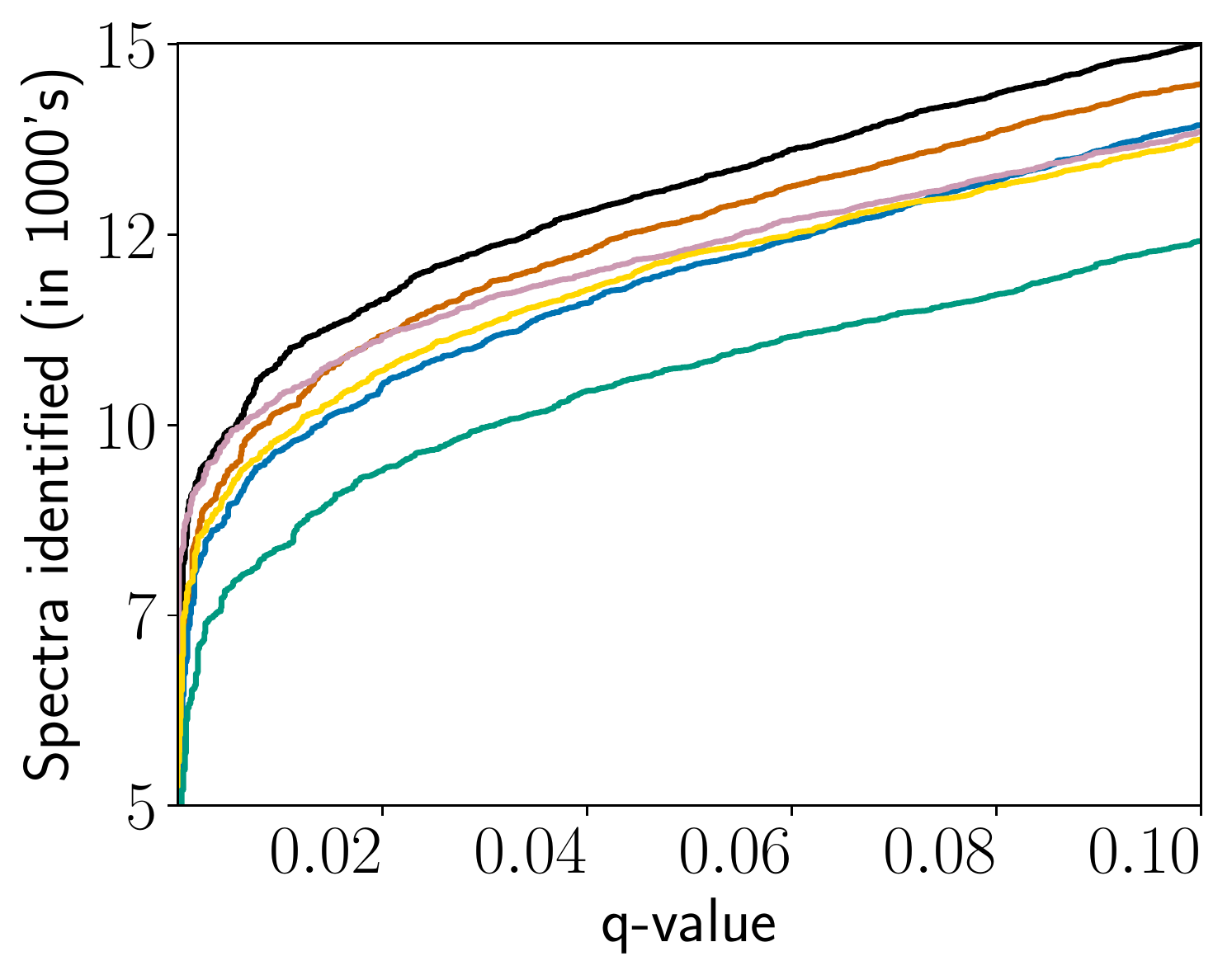}}
  \subfigure[Yeast-3]{\includegraphics[trim=0.45in 0.0in 0.0in 0.05in,
    clip=true,scale=0.30]{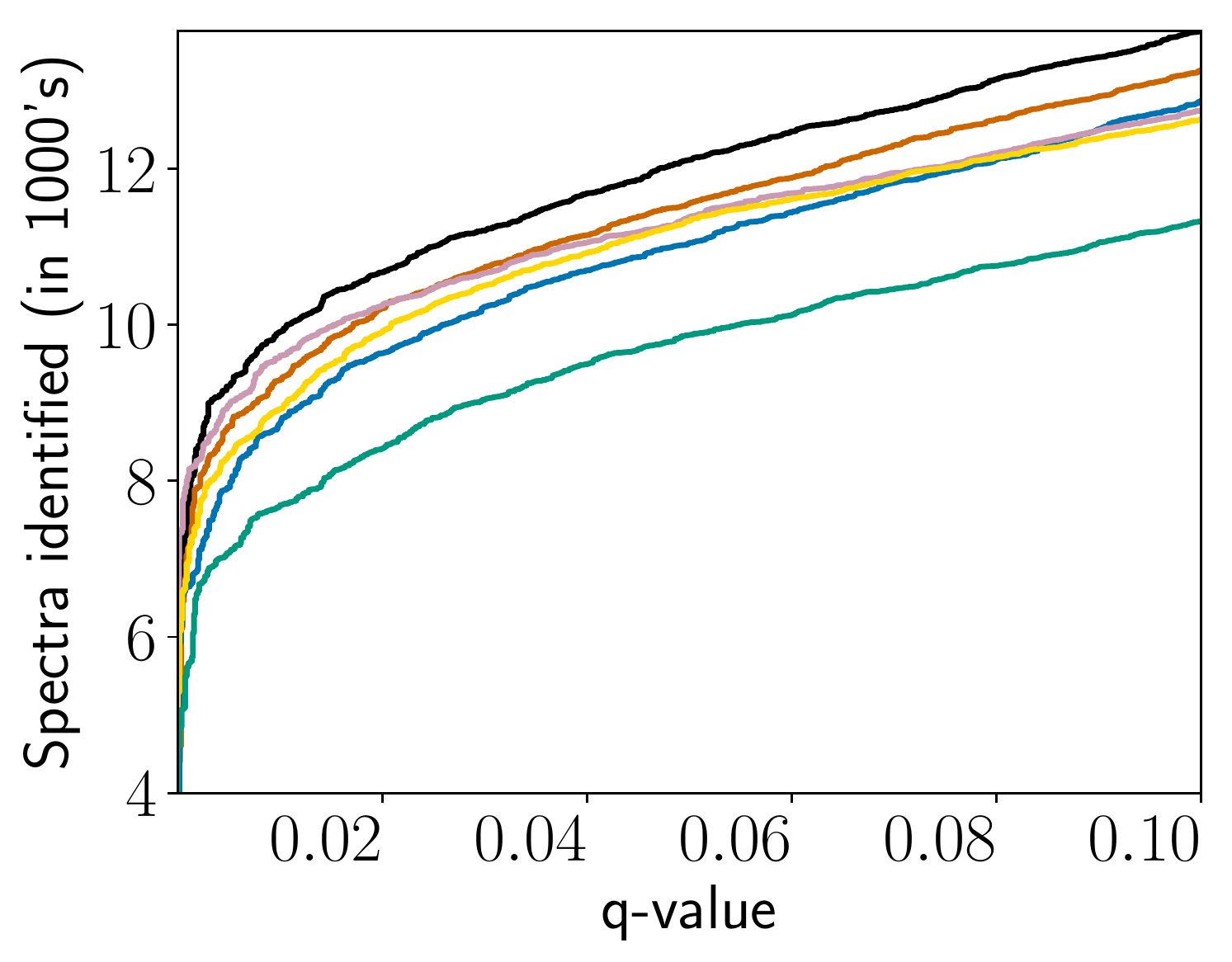}}
  \subfigure[Yeast-4]{\includegraphics[trim=0.45in 0.0in 0.0in 0.05in,
    clip=true,scale=0.30]{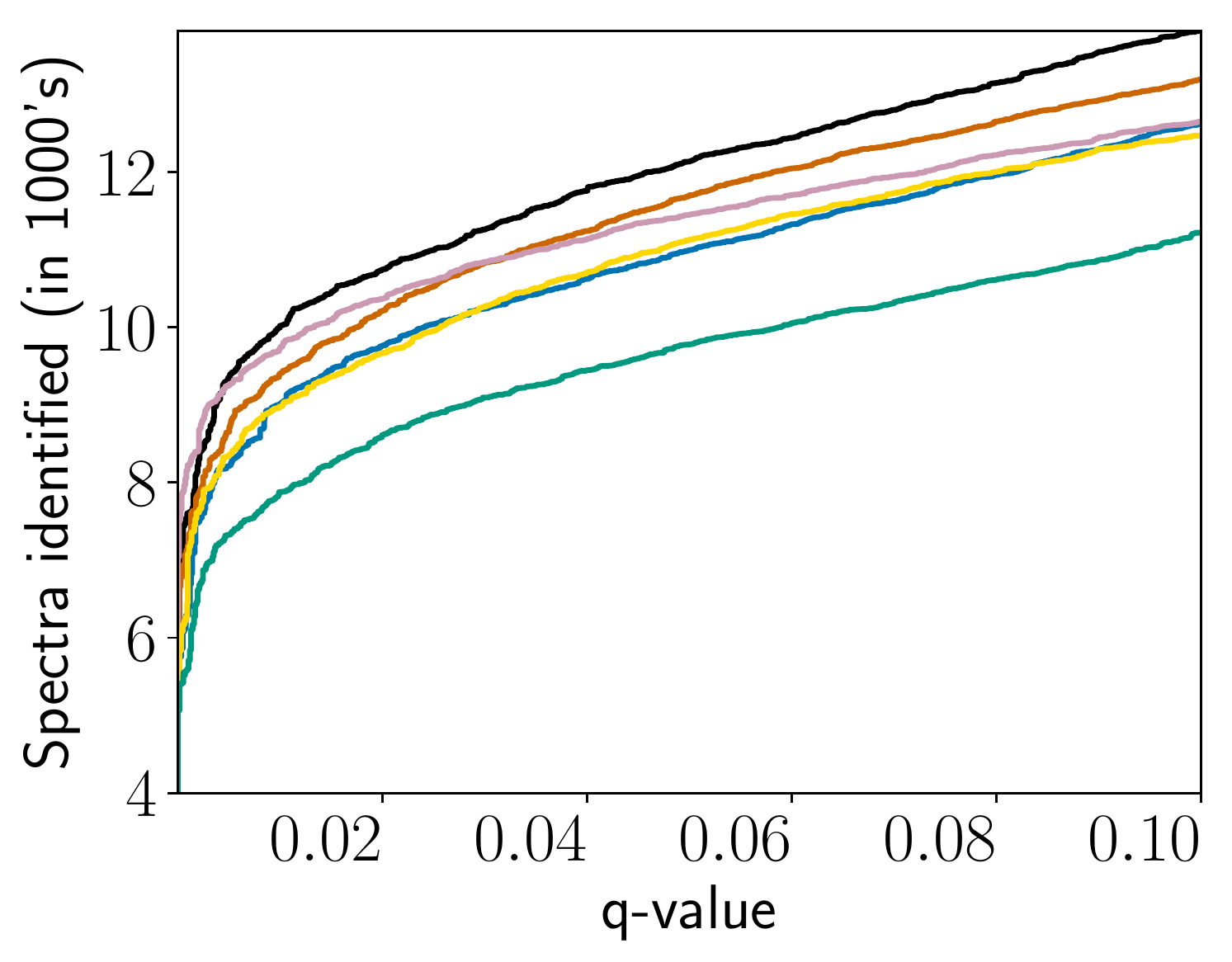}}
  \caption{{\small Database search accuracy plots measured by
    $q$-value versus number of spectra identified for worm
    (\emph{C. elegans}) and yeast (\emph{Saccharomyces cerevisiae})
    datasets.  All methods are run with as equivalent settings as
    possible.  The Didea charge-varying model was used to score PSMs,
    with ``Didea'' denoting the model trained per charge state using
    the concave framework described in Section~\ref{section:veLse} and
    ``Didea-0'' denoting the model from~\cite{singh2012-didea-uai}
    trained using a grid search.  DRIP, another DBN-based scoring
    function, was run using the generatively learned parameters
    described in~\cite{halloran2016dynamic}.}
}
  \label{fig:absRanking}
\end{figure*}
The benchmark datasets and search settings used to recently evaluate the DRIP
Fisher kernel in~\cite{halloran2017gradients} are adapted in this
work.  The charge-varying Didea model
(which integrates over multiple charge states, further described
in Appendix~\ref{section:chargeIntegration}) with concave emissions (described in
Section~\ref{section:veLse}) was used to
score and rank database peptides.  Concave Didea parameters were
learned using the high-quality
PSMs used to generatively train the DRIP model
in~\cite{halloran2016dynamic} and gradient ascent.  Didea's newly trained
database-search scoring function is benchmarked against the Didea
model from~\cite{singh2012-didea-uai} trained using a costly grid search
for a single parameter (denoted as ``Didea-0'') and four other
state-of-the-art scoring algorithms: DRIP, MS-GF+, XCorr $p$-values,
and XCorr.  

DRIP searches were conducted using the DRIP Toolkit and the
generatively trained parameters described
in~\cite{halloran2014uai-drip, halloran2016dynamic}.  MS-GF+, one of
the most accurate search algorithms in wide-spread use, was run using
version 9980, with PSMs ranked by E-value.  XCorr
and XCorr $p$-value scores were collected using Crux v2.1.17060.  All
database searches were run using a $\pm 3.0 \thomson$ mass tolerance,
XCorr flanking peaks not allowed in Crux searches, and all search
algorithm settings otherwise left to their defaults.  Peptides were
derived from the protein databases using trypsin cleavage rules
without suppression of proline and a single fixed carbamidomethyl
modification was included.  

The resulting database-search accuracy plots are
displayed in Figure~\ref{fig:absRanking}.  The trained Didea model
outperforms all competitors across all presented datasets;
compared to highly accurate scoring algorithms DRIP and MS-GF+, the
trained Didea scoring function identifies $16\%$ more spectra than
DRIP and $17.4\%$ more spectra than MS-GF+, at a strict FDR of $1\%$
averaged over the presented datasets.  This high-level performance is
attributable to the expanded and efficient parameter learning
framework, which greatly improves upon the limited parameter learning
capabilities of the original Didea model, identifying $9.8\%$ more
spectra than Didea-0 at a strict FDR of $1\%$
averaged over the presented datasets.

\subsection{Conditional Fisher kernel for improved discriminative
  analysis}\label{section:conditionalFisher}
Facilitated by Didea's effective parameter learning framework, we look to
leverage gradient-based PSM information to aid in discriminative
postprocessing analysis.  We utilize the same set of features as
the DRIP Fisher kernel~\cite{halloran2017gradients}.  However, in order
to measure the relative utility of the gradients under study, we
replace the DRIP log-likelihood gradients with Didea gradient
information.
These features are used to train an SVM classifier,
Percolator~\cite{kall:semi-supervised}, which recalibrates PSM scores
based on the learned decision boundary between input targets and
decoys.  Didea's resulting conditional Fisher kernel is
benchmarked against the DRIP Fisher kernel and the previously
benchmarked scoring algorithms using their respective standard
Percolator features sets.  

DRIP Kernel features were computed using the
customized version of the DRIP Toolkit
from~\cite{halloran2017gradients}.  MS-GF+ Percolator features were
collected using \texttt{msgf2pin} and XCorr/XCorr $p$-value features
collected using Crux.  For the resulting postprocessing results, the trained Didea
scoring function outperforms all competitors, identifying $12.3\%$ more spectra than DRIP and
$13.4\%$ more spectra than MS-GF+ at a strict FDR of
$1\%$ and averaged over the presented datasets.  The full panel of
results is displayed in Appendix~\ref{section:conditionalFisher}.
Compared to DRIP's log-likelihood gradient features, the conditional
log-likelihood gradients of Didea contain much richer PSM information, thus
allowing Percolator to better distinguish target from decoy PSMs for
much greater recalibration performance.

\subsection{Optimized exact sum-product inference for improved
  Didea runtime}\label{section:fasterResults}
Implementing the speedups to exact Didea sum-product inference
described in Section~\ref{section:faster}, we benchmark the optimized
search algorithm using $1,000$ randomly sampled spectra (with charges
varying from 1+ to 3+) from the Worm-1 dataset and averaged
database-search times (reported in wall clock time) over 10 runs.
The resulting runtimes are listed in Table~\ref{table:speedups}.  DRIP was
run using the DRIP Toolkit and XCorr $p$-values were collected using
Crux v2.1.17060.  All benchmarked search algorithms were run on the
same machine with an Intel Xeon E5-2620 and 64GB RAM.  The described
optimizations result in a $64.2\%$ runtime improvement, and brings
search time closer to less accurate, but faster, search algorithms.
\begin{table}[h]
\begin{center}
{\small
\tabcolsep=0.11cm
\begin{tabular}{|c|c|c|c|c|}
\hline
& Didea-0 & Didea Opt. & XCorr $p$-values & DRIP\\\hline
runtime & 19.1175 & 6.8535 & 2.2955 & 143.4712\\\hline
\end{tabular}
}
\caption{Database search runtimes per spectrum, in seconds, searching
  1,000 worm spectra randomly sampled from the Worm-1 dataset.
  ``Didea-0'' is the implementation of Didea used
  in~\cite{singh2012-didea-uai} and ``Didea Opt'' is the
  speed-optimized implementation described herein.  All reported
  search algorithm runtimes were averaged over 10 runs.}
\label{table:speedups}
\end{center}
\end{table}

\section{Conclusions and future work}\label{section:conclusions}
In this work, we've derived a widely applicable class of Bayesian
network emission distributions, CVEs, which naturally generalize the
convex log-sum-exp function and carry important theoretical properties
for parameter learning.  Using CVEs, we've substantially improved the
parameter learning capabilities of the DBN scoring algorithm, Didea,
by rendering its conditional log-likelihood concave with respect to a
large set of learnable parameters.  Unlike previous DBN parameter
learning solutions, which only guarantee convergence to a local
optimum, the new learning framework thus guarantees global
convergence.  Didea's newly trained database-search scoring function 
significantly outperforms all state-of-the-art scoring algorithms
on the presented datasets.  With efficient parameter learning in hand,
we derived the gradients of Didea's conditional
log-likelihood and used this gradient information in the feature space
of a kernel-based discriminative postprocessor.  The resulting
conditional Fisher kernel once again outperforms the
state-of-the-art on all presented datasets, including a highly
accurate, recently proposed Fisher kernel.  Furthermore, we
successively optimized Didea's message passing
schedule, leading to DBN analysis times two orders of magnitude faster than other leading
DBN tools for MS/MS analysis.  Thus, the presented
results improve upon all aspects of state-of-the-art DBN analysis for
MS/MS.  Finally, using the new learning framework, we've proven that
Didea is proportionally lower bounds the widely used XCorr scoring
function.

There are a number of exciting avenues for future work.  Considering
the large amount of PSM information held in the gradient space of
Didea's conditional log-likelihood, we plan on pursuing kernel-based
approaches to peptide identification using the Hessian of the scoring
function.  This is especially exciting given the high degree of recalibration
accuracy provided by Percolator, a kernel-based post-processor.  Using a variational
approach, we also plan on
investigating parameter learning options for XCorr given the Didea
lower bound and the concavity of Didea's parameterized scoring function.
Finally, in perhaps the most ambitious plan for future work, we plan
to further build upon Didea's parameter learning framework by learning
distance matrices between the theoretical and observed spectra.  Such
matrices naturally generalize the class of CVEs derived herein.

\noindent {\bf Acknowledgments}: This work was supported by the
National Center for Advancing Translational Sciences (NCATS), National
Institutes of Health, through grant UL1 TR001860.
\bibliographystyle{plain}
\setcitestyle{numbers, open={[}, close={]}}
\bibliography{dideaLearning_arxiv}

\appendix
\section{Proofs for Convex Virtual Emissions}\label{appendix:cveProof}
Consider the log-posterior
probability
{\small
\begin{equation}
\log p(h | E) = \log \frac{p(h,
  E)}{p(E)} = \log \frac{p(h,
  E)}{\sum_{\bar{h} \in H} p(\bar{h},E)} = \log p(h, E) - \log \sum_{h \in H}
  p(h)p(E | h),\label{eq:posteriorProb}
\end{equation}
}
as is computed in Didea.  Assume $p(h)$ and $p(E
| h)$ are non-negative
probability distributions and that the
emission density $p(E | h)$ is parameterized by $\theta$ which we'd like
to learn.  Applying virtual evidence for such models, $p(E | h)$
need not be normalized for posterior inference, Viterbi inference, and
comparative inference between sets of observations (discussed at
length in~\cite{halloranThesis2016}).

Assume that there is a parameter $\theta_h$ to
be learned for every hypothesis of latent variables, though if we have
fewer parameters, parameter estimation becomes strictly
easier.  We make this parameterization explicit by denoting the
emission distributions of interest as $p_{\theta_h}(E | h)$.  We first
look for functions $f_h(\theta_h) = p_{\theta_h}(E | h)$
which render the log-likelihood convex,
\begin{align}
\log p(E) =& \log \sum_{h \in H} 
p(h)p_{\theta_h}(E|h) = \log \sum_{h \in H} c_h
f_h(\theta_h),\label{equation:negLogPosterior}
\end{align}
where $c_h = p(h)$ are nonnegative constants with regards to the
parameters of interest.  Assume that $f_h(\cdot)$ is smooth on $\cR$
for all $h \in
H$.
\begin{theorem}\label{theorem:generalizedLse}
The convex functions of the form $\log \sum_{h \in H} c_h f_h(\theta_h)$, such
that $(f_h'(\theta_h))^2  - f_h''(\theta_h)f_h(\theta_h) = 0$, are $\log
\sum_{h \in H} c_h f_h (\theta_h) = \log \sum_{h \in H} c_h \alpha_h
e^{\beta_h\theta_h}$, where $\alpha_h$ and $\beta_h$ are constants
uniquely determined by initial conditions.
\end{theorem}
\begin{proof}
In order to ensure convexity of
Equation~\ref{equation:negLogPosterior}, it is necessary and
sufficient that $\nabla_{\theta}^2 \log \sum_{h \in H} c_h
f_h(\theta_h)\succeq 0$.  We
thus have the following for the gradient
\begin{align*}
\nabla_{\theta} \log \sum_{h \in H} c_h
f_h(\theta_h) = & \frac{1}{\sum_{h \in H} c_hf_h(\theta_h)}
\begin{bmatrix} 
  c_1 f_1'(\theta_1)\\ 
  \vdots \\
  c_{\hCard} f_{\hCard}'(\theta_{\hCard})
\end{bmatrix}.
\end{align*}
Letting $Z = \sum_{h} c_h f_h(\theta_h)$, we have
\begin{align*}
\frac{\delta \log p_{\theta} (h | E)}{\delta \theta_i \delta
  \theta_j}
= & \begin{cases}
\frac{c_if_i''(\theta_i)Z - (c_if_i'(\theta_i))^2}{Z^2} & \mbox{ if }
i = j\\
\frac{- c_if_i'(\theta_i) c_j f_j' (\theta_j)}{Z^2} & \mbox{ if } i
\neq j
\end{cases}
\end{align*}
Letting $a = \begin{bmatrix}c_1 f''_1(\theta_1) & \dots & c_{\hCard}
  f''_{\hCard}(\theta_{\hCard})\end{bmatrix}^T$ and $b
= \begin{bmatrix}c_1 f'_1(\theta_1) & \dots & c_{\hCard}
  f'_{\hCard}(\theta_{\hCard})\end{bmatrix}^T$, we may thus write the
Hessian as
\begin{align}\label{equation:compactHessian}
\nabla_{\theta}^2 \log \sum_{h \in H} c_h
f_h(\theta_h) = & \frac{ \mbox{diag} (a) }{Z}
- \frac{1}{Z^2}bb^T.  
\end{align}

Equation~\ref{equation:compactHessian} is positive semi-definite if and only
if, for all $x \in \cR ^n$,
\begin{align}
x^T \nabla_{\theta}^2 \log \sum_{h \in H} c_h
f_h(\theta_h) x \geq & 0 \nonumber\\
Z^2 x^T \nabla_{\theta}^2 \log \sum_{h \in H} c_h
f_h(\theta_h) x \geq & 0\nonumber\\
x^T ( Z \mbox{diag}(a) - bb^T) x \geq & 0 \nonumber\\
x^T Z \mbox{diag}(a)x \geq &  x^T bb^T x\nonumber\\
( \sum_i c_i f_i(\theta_i) )( \sum_i c_i f_i''(\theta_i)x_i^2 )  \geq & ( \sum_i c_i
f_i' (\theta_i) x_i )^2.\label{equation:hessianBound}
\end{align}


Letting $l, u, v$ be vectors with components 
$l_i = x_i \sqrt{c_if''_i(\theta_i)}, u_i = x_i
\frac{c_if'_i(\theta_i)}{\sqrt{c_if_i(\theta_i)}}, v_i =
\sqrt{c_if_i(\theta_i)}$, we require
\begin{align}\label{equation:generalizedLse}
(l^Tl)(v^Tv) \geq& (u^Tv)^2 .
\end{align}
Note that, by the non-negativity of $c_h$ and $f_h(\theta_h)$, $v_i$
is real and the quantify $l^Tl$ is always real.  When $l = u$, the
bound in Equation~\ref{equation:generalizedLse} is guaranteed to hold
by the Cauchy-Schwarz inequality.  Thus, the
log-probability of evidence is convex when
\begin{align}
l_i =& u_i \nonumber\\
\Rightarrow x_i \sqrt{c_i f_i''(\theta_i)} =& x_i \frac{c_i
  f_i'(\theta_i)}{\sqrt{c_i f_i(\theta_i)}} \nonumber\\
\Rightarrow f_i''(\theta_i)f_i(\theta_i) =&
(f_i'(\theta_i))^2.\label{equation:lseOde}
\end{align}
Equation~\ref{equation:lseOde} is an autonomous, second-order,
nonlinear ordinary differential equation (ODE), the solution of which
leads us to the following generalization of the commonly encountered
LSE convex function.

To simplify notation, let $t= \theta_i$ and $y(t) = f_i(t)$, where
we drop the independent variable when it is understood.  Thus, we are
looking for $y$ such that 
\begin{align}
y'' =& \frac{(y')^2}{y}.\label{equation:rewrittenOde}
\end{align}
Let $v(t) = y'(t)$, $w(y)
= v(t(y))$, and $\dot{w} = \frac{dw}{dy}$.  We thus have
\begin{align*}
v' =& \frac{v^2}{y},\\
\dot{w}(y) =& \frac{dw}{dy}(y) =
\frac{dv}{dt}\frac{dt}{dy} \bigg|_{t(y)} = \frac{v'}{y'} \bigg|_{t(y)}
= \frac{v'}{v} \bigg|_{t(y)}.
\end{align*}
Using Equation~\ref{equation:rewrittenOde}, we have
\begin{align*}
\dot{w}(y) =& \frac{v'(t(y))}{w(y)}= \frac{v^2(t(y))}{w(y)y} =
\frac{w^2(y)}{w(y)y} = \frac{w(y)}{y}.
\end{align*}
Solving this ODE using seperation of variables, we have
\begin{align*}
\ln w =& \ln y + d_0\\
\Rightarrow w =& \exp (\ln{y} + d_0)= e^{d_0}y  = d_1y,
\end{align*}
where $d_0$ is a constant of integration.  To solve for $y$, we have
\begin{align*}
w(y(t)) = v(t) = y' = d_1 y.
\end{align*}
As before, we solve $y' = d_1y$ using seperation of variables, giving
us 
\begin{align*}
\ln y =& d_1 t + d_2\\
\Rightarrow y(t) =& d_3e^{d_1t},
\end{align*}
where $d_1$ and $d_3$ are constants uniquely determined by initial
conditions.  Returning to our earlier notation, the solution to
Equation~\ref{equation:lseOde} is thus $f_i(\theta_i) = d_3 e^{d_1
\theta_i}$.  
Letting $\beta_i = d_1$ and
$\alpha_i = d_3$ completes the proof.
\end{proof}
\begin{corollary}
For convex $\log p(E) = \log \sum_{h \in H} p(h)p_{\theta_h}(E | h)$ such that
$(p_{\theta_h}'(E|h))^2 - p_{\theta_h}''(E|h)p_{\theta_h}(E|h) = 0$, the
log-posterior $\log p_{\theta}(h | E)$ is concave in $\theta$.
\end{corollary}
\begin{proof}
From Theorem~\ref{theorem:generalizedLse}, $p_{\theta_h}(E | h) = \alpha_h
e^{\beta_h\theta_h}$ and we have
\begin{align*}
\log p_{\theta}(h | E) & = \log p_{\theta_h}(h, E) - \log \sum_{h \in H}
  p(h)p_{\theta_h}(E | h)\\
& = \log p(h) p_{\theta_h}(E | h) - \log \sum_{h \in H}
  p(h)p_{\theta_h}(E | h)\\
& = \log p(h) \alpha_he^{\beta_h\theta_h} - \log \sum_{h \in H}
  p(h)p_{\theta_h}(E | h)\\
& = \log p(h) \alpha_h + \beta_h\theta_h - \log \sum_{h \in H}
  p(h)p_{\theta_h}(E | h).
\end{align*}
With respect to $\theta$, $\beta_h\theta_h$ is affine, $ - \log \sum_{h \in H}
  p(h)p_{\theta_h}(E | h)$ is concave, and the remaining term is
  constant.
\end{proof}

\section{Analysis of CVEs in Didea}\label{section:veLse}
In order to analyze Didea's scoring function,
define boolean vectors length $\maxth$ $\bion, \yion$  such that, for
the set of b-ions $\beta_{\pep}=\{\cup_{i=1}^{l-1} \{b(m(x_{1:i}),1)\}
\}$ and y-ions $\upsilon_{\pep}=\{\cup_{i=0}^{l-1}\{y(m(x_{i+1:l}),1)
\} \}$ of $x$, and $0 \leq j \leq \maxth$ we have 
\begin{align*}
\bion(j) =& \indicator_{\{ j \in \beta \}}, \;\;\;
\yion(j) =
\indicator_{\{ j \in \upsilon \}}.
\end{align*}
We note that computing Didea scores as detailed in the sequel would be
much more slower than computing Didea scores using sum-product
inference.  However, the compact description of Didea's scoring
function allows much easier analysis in
Appendices~\ref{section:dideaGradients} and~\ref{section:xcorrbound}.

Recall the CVE used in the main paper,
$f_{\theta_\tau}(\binspec (i)) = e^{\theta_{\tau} \binspec (i)}$.
Under this new emission distribution, Didea's scoring
function may thus be compactly written as
{\small
\begin{align}
\dideaScore_{\lambda}(\binspec, \pep) 
&=\log{\distp ( \pep, \procObs | \dideaShift_0=0)}
-\log{ \sum_{\tau}\distp(\pep,
  \procObs_{\dideaShift}|\dideaShift_0=\tau) } \nonumber\\
&= \sum_{t=1}^{l} ( \log \cve{0}(b_t) + 
\log \cve{0}(y_{n-t}) ) - \log{
  \sum_{\tau} \exp{\sum_{t=1}^{l} ( \log \cve{\tau}(b_t) + 
\log \cve{\tau}(y_{n-t}) ) }} \nonumber\\
&= \theta_0\bion^T \binspec + \theta_0\yion^T \binspec -\log{ 
  \sum_{\tau} \exp ( {\theta_{\tau}\bion^T \binspec_{\tau} +
  \theta_{\tau}\yion^T \binspec_{\tau}  ) }}= \theta_0(\bion + \yion)^T \binspec -\log{ 
  \sum_{\tau} \exp {[\theta_{\tau}(\bion + \yion)^T
  \binspec_{\tau}  ]}}.\label{equation:compactDidea}
\end{align}
}
\subsection{Gradients of CVEs in Didea}\label{section:dideaGradients}
Letting $h_{\tau}(x,s) = (\bion +
\yion)^Ts_{\tau}$, the gradient of this new conditional log-likelihood
has elements
{\small
\begin{align}
\frac{\delta}{\delta \theta^{\tau}}
\psi_{\theta}(s,x) \bigg\rvert_{\tau=0}&= h_0(x,s)-\frac{1}{\sum_{0} e^{\theta_{0}
    h_0(x,s)}} \sum_{0} h_0(x,s) e^{\theta_{0}
  h_0(x,s)} \label{eq:dideaDerivativeA}\\
\frac{\delta}{\delta \theta_{\tau}} \psi_{\theta}(s,x)
\bigg\rvert_{\tau\neq0}&= - 
\frac{1}{\sum_{\tau} e^{\theta_{\tau}
    h_{\tau}(x,s)}} \sum_{\tau} h_{\tau}(x,s) e^{\theta_{\tau}
  h_{\tau}(x,s)}.\label{eq:dideaDerivativeB}
\end{align}
}
Given $N$ i.i.d.\ training PSMs $\{ (s^1,x^1), (s^2, x^2),
\dots, (s^N, x^N) \}$, we need only run sum-product inference once to
cache the values $\{\cup_{i=1}^N
\{h_{-\dideaMaxShift}(s^i,x^i), \dots , h_{0}(s^i,x^i), \dots,
h_{\dideaMaxShift}(s^i,x^i) \}\}$ for extremely fast gradient based
learning.
\section{PROOF OF XCORR UPPER BOUND FOR DIDEA'S SCORING FUNCTION}
\subsection{Proof of Didea lower bound for the XCorr scoring function}\label{section:xcorrbound}
\begin{theorem}\label{theorem:xcorrbound}
Assume the PSM scoring function $\dideaScore(s,x)$ is that of Didea
under the emission function $f_{\theta_{\tau}}(s(i))$ with uniform
weights $\theta_{i} = \theta_{j}$, for $i,j \in [-L, L]$.  Then
$\dideaScore(s,x) \leq \bigo(\mbox{XCorr}(\obsSpec, \pep))$.
\end{theorem}
\begin{proof}
Recall that, for theoretical spectrum $u$, XCorr is computed as
{\small
\begin{align*}
\mbox{XCorr}(\obsSpec, \pep) &= \theoVector^Ts-
\frac{1}{2L+1}\sum_{\tau=-L}^{L}\theoVector^Ts_{\tau} =
\theoVector^T(s- \frac{1}{2L + 1}\sum_{\tau=-L}^{L}s_{\tau} )=
\theoVector^Ts'.
\end{align*}}
Let $\lambda = \theta_{i}$ for $i \in [-L, L]$.  From Didea's scoring
function, we have 
{\small
\begin{align*}
\dideaScore(\obsSpec, \pep) 
&=\log{\distp ( \pep, \procObs, \dideaShift_0=0)}
-\log{\sum_{\tau  = -\dideaMaxShift}^{\dideaMaxShift}
  \distp(\dideaShift_0=\tau)
  \distp(\pep, \procObs_{\dideaShift}|\dideaShift_0=\tau)}= \log{\distp ( \pep, \procObs, \dideaShift_0=0)}
-\log{ \expect [\distp(\pep, \procObs_{\dideaShift}|\tau)] },
\end{align*}
}
so that, by Jensen's inequality,
\begin{align}
\dideaScore(\obsSpec, \pep) & \leq \log{\distp ( \pep, \procObs, 
  \dideaShift_0=0)}
-\expect [\log \distp(\pep,
\procObs_{\dideaShift}|\tau)]\label{eq:dideaJensen}.
\end{align}

The right-hand side of~\ref{eq:dideaJensen}, which we'll denote as
$g(s,x)$, is 
{\small
\begin{align*}
g(s,x) &=\log { \frac{1}{\tauCard} \distp ( \pep, \procObs | \dideaShift_0=0) }
-\expect [\log \distp(\pep,
\procObs_{\dideaShift}|\dideaShift_0=\tau)]\\
 &= -\log{\tauCard} + \lambda (\bion + \yion)^T s-\sum_{\tau = -\dideaMaxShift}^{\dideaMaxShift} \distp (\dideaShift_0=0) \log{ e^{\lambda
    (\bion+\yion)^T s_{\tau}}}\\
&= -\log{\tauCard}+ \lambda(\bion + \yion)^T s
-\frac{\lambda}{\tauCard}\sum_{\tau = -\dideaMaxShift}^{\dideaMaxShift}
(\bion+\yion)^T s_{\tau}.
\end{align*}
}
Letting $\theoVector = \bion + \yion$, we have
{\small
\begin{align*}
g(s,x) &= -\log{\tauCard}+ \lambda\theoVector^T s
-\frac{\lambda}{\tauCard}\sum_{\tau = -\dideaMaxShift }^{\dideaMaxShift} \theoVector^T
s_{\tau}= -\log{\tauCard}+ \lambda\theoVector^T ( s
-\frac{1}{\tauCard}\sum_{\tau = -\dideaMaxShift}^{\dideaMaxShift}
s_{\tau})= -\log{\tauCard}+ \lambda\mbox{XCorr}(s,x).
\end{align*}
}
\\
$\Rightarrow \dideaScore(\obsSpec, \pep) \leq g(s,x)=-\log{\tauCard}+
\lambda\mbox{XCorr}(s,x)$
\end{proof}
\section{Charge varying spectra}\label{section:chargeIntegration}
In practice, observed spectra exhibit higher charged fragment
ions.  In order to account for these new fragmentation peaks while
still keeping the scoring function well calibrated (i.e., keeping
higher charged PSMs comparable in range to lower charged PSMs), a
Didea charge varying model is introduced
in~\cite{singh2012-didea-uai}.  In this model, a global variable
switches between two separate models: the singly charged model and a
model which considers both single and
double charged fragment ions.  The latter model contains a charge
variable in every frame which is hidden and integrated over.  This
effectively averages the contribution between the differently charge
b-and y-ion pairs.  Finally, the contribution between the two separate
charged models is averaged (per frame).  The posterior of $\tau_0 = 0$
remains Didea's PSM score in this setting.    Further details of the
model's scoring function may be
found in~\cite{singh2012-didea-uai}.

\section{Conditional Fisher kernel for improved discriminative
  analysis}\label{section:conditionalFisher}
We leverage Didea's gradient-based PSM information to aid in
discriminative postprocessing analysis.  We utilize the same set of
features as the DRIP Fisher kernel~\cite{halloran2017gradients} where,
to measure the relative utility of the gradients under study, the DRIP
log-likelihood gradients are replaced with Didea gradient
information (derived in Section~\ref{section:dideaGradients}).
These features are used to train an SVM classifier,
Percolator~\cite{kall:semi-supervised}, which recalibrates PSM scores
based on the learned decision boundary between input targets and
decoys.  The resulting Didea conditional Fisher kernel is
benchmarked against the DRIP Fisher kernel and the scoring algorithms
benchmarked in the main paper (using their respective standard
Percolator features sets).  DRIP Kernel features were computed using a
customized version of the DRIP Toolkit, provided by the authors of
~\cite{halloran2017gradients}.  MS-GF+ Percolator features were
collected using \texttt{msgf2pin} and XCorr/XCorr $p$-value features
collected using Crux.  The resulting postprocessing results are
displayed in Figure~\ref{fig:percolatorAbsRanking}.
\begin{figure*}[htbp!]
  \centering
  \subfigure{\raisebox{8.0mm}{\includegraphics[trim=3.8in 1.0in 0.3in 2.3in,
    clip=true,scale=0.7]{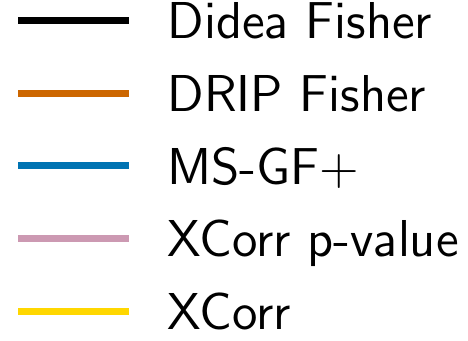}}}
  \subfigure[Worm-1]{\includegraphics[trim=0.0in 0.0in 0.0in 0.05in,
    clip=true,scale=0.30]{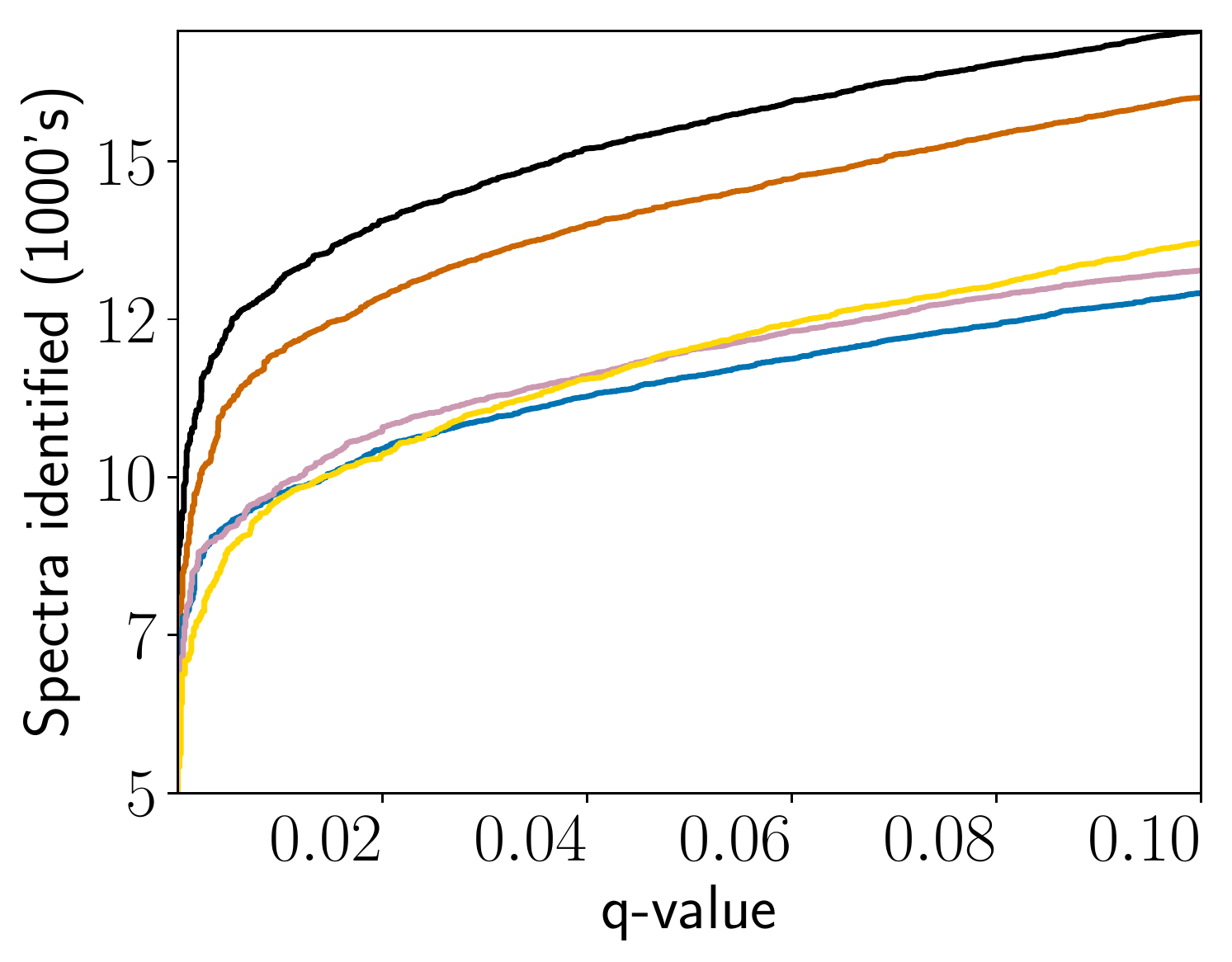}}
  \subfigure[Worm-2]{\includegraphics[trim=0.45in 0.0in 0.0in 0.05in,
    clip=true,scale=0.30]{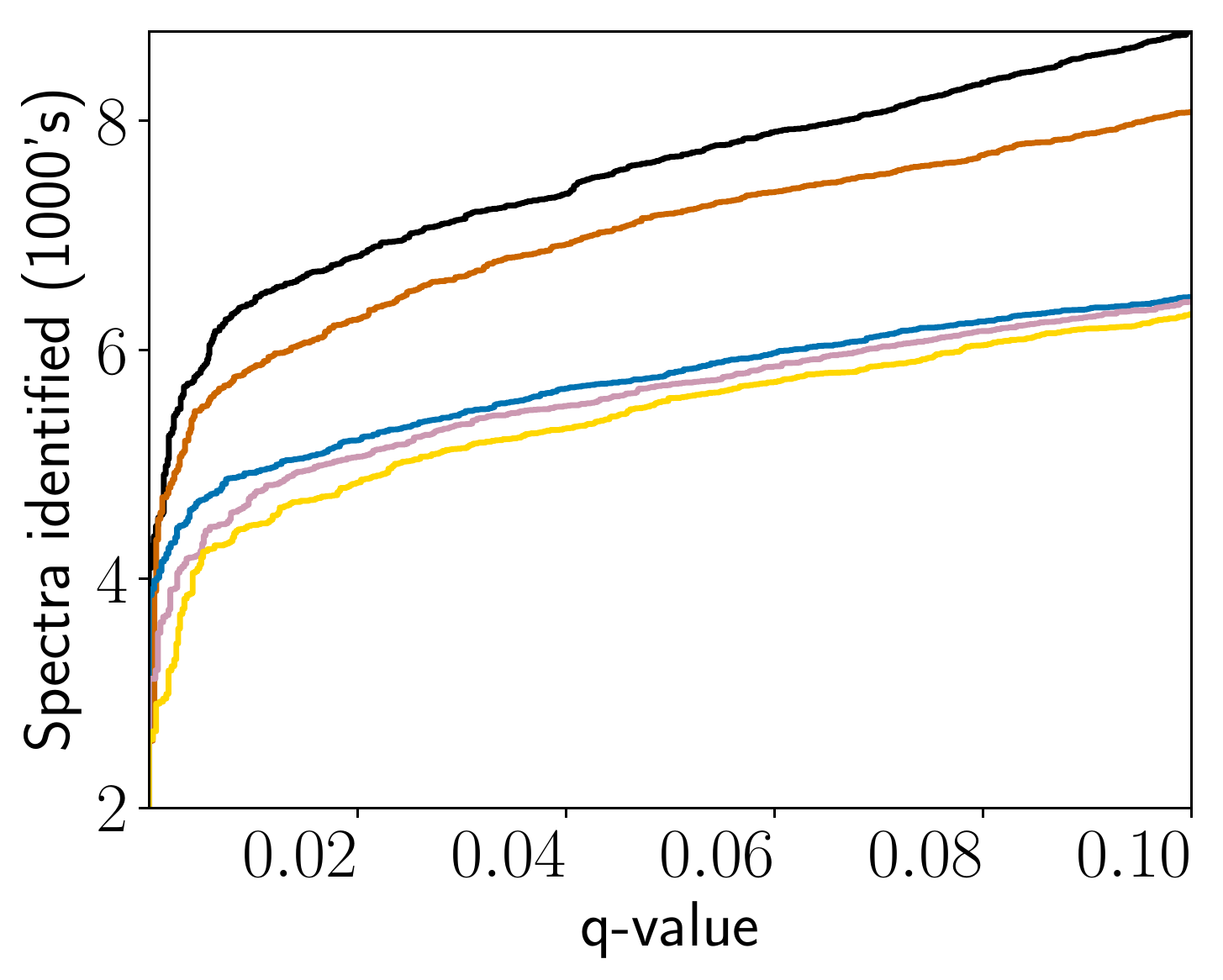}}
  \subfigure[Worm-3]{\includegraphics[trim=0.0in 0.0in 0.0in 0.05in,
    clip=true,scale=0.3]{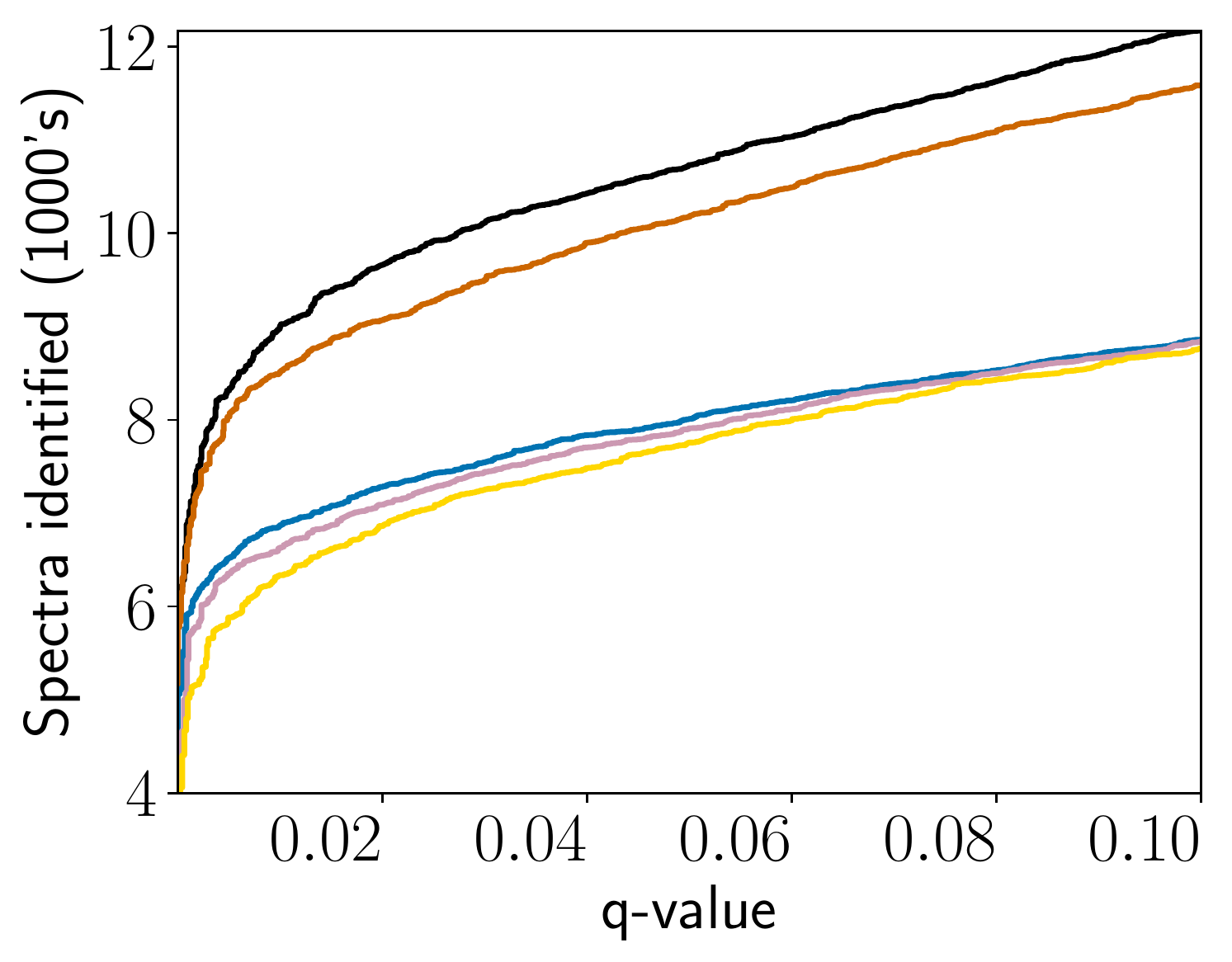}}
  \subfigure[Worm-4]{\includegraphics[trim=0.45in 0.0in 0.0in 0.05in,
    clip=true,scale=0.3]{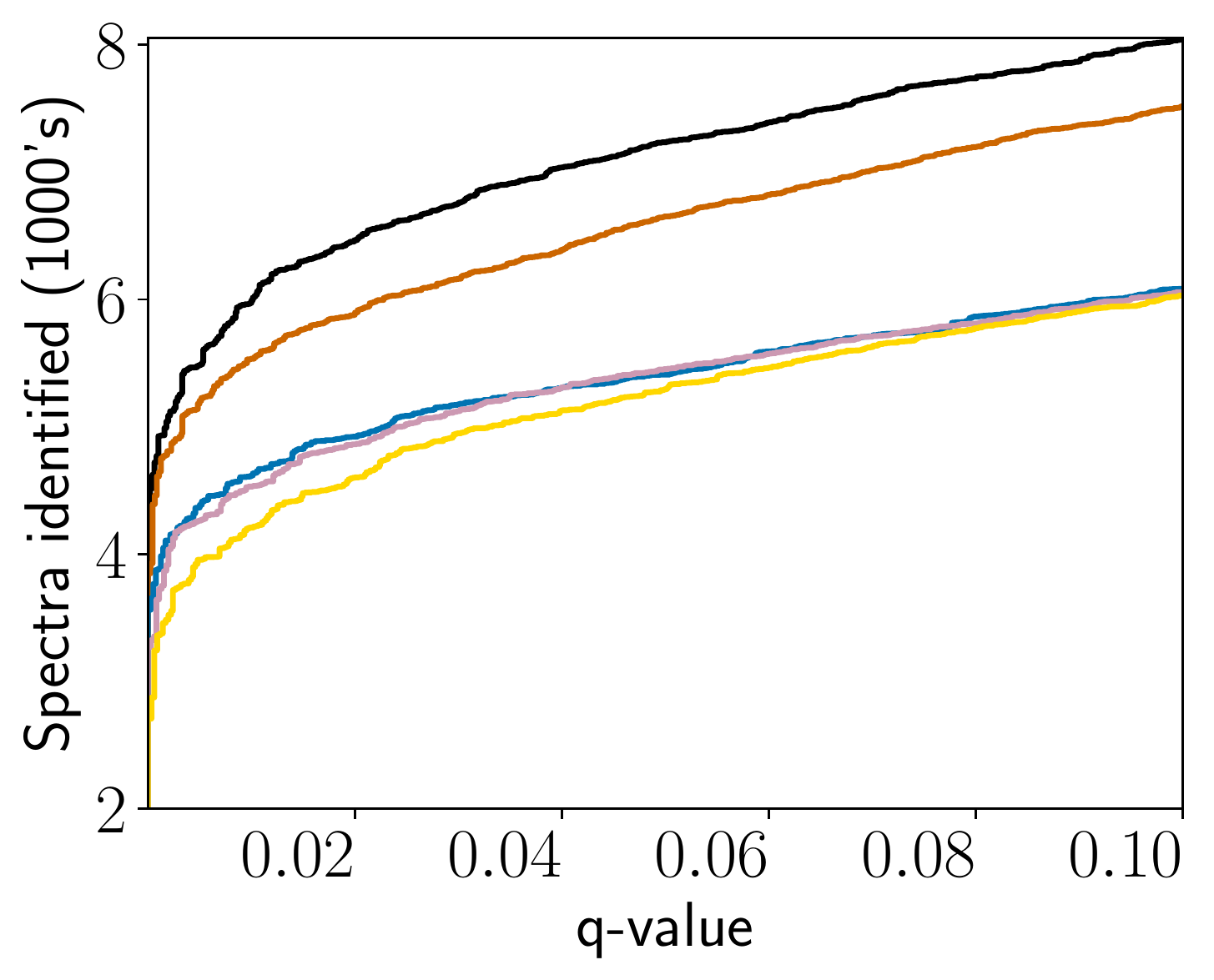}}
  \subfigure[Yeast-1]{\includegraphics[trim=0.45in 0.0in 0.0in 0.05in,
    clip=true,scale=0.3]{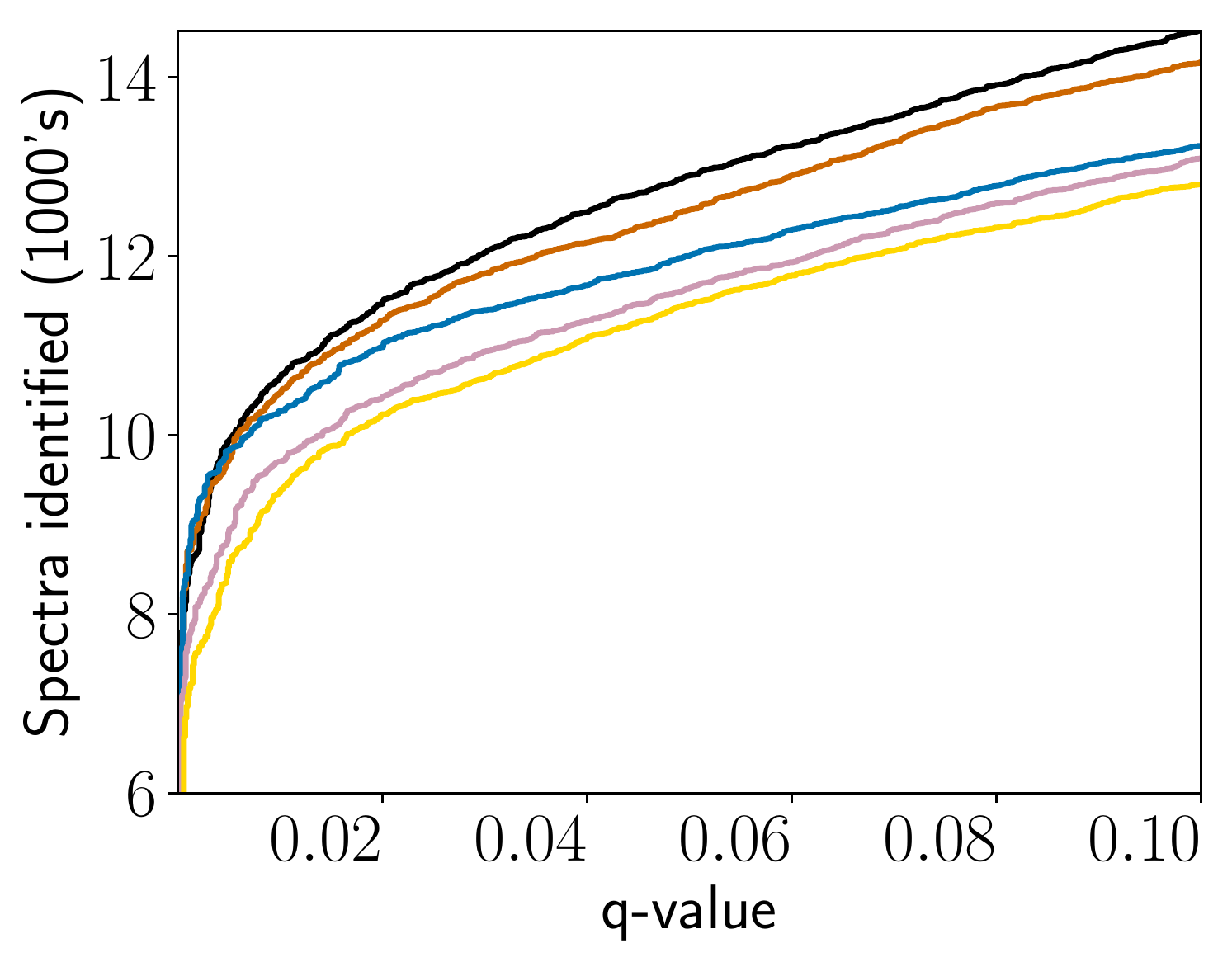}}
  \subfigure[Yeast-2]{\includegraphics[trim=0.0in 0.0in 0.0in 0.05in,
    clip=true,scale=0.3]{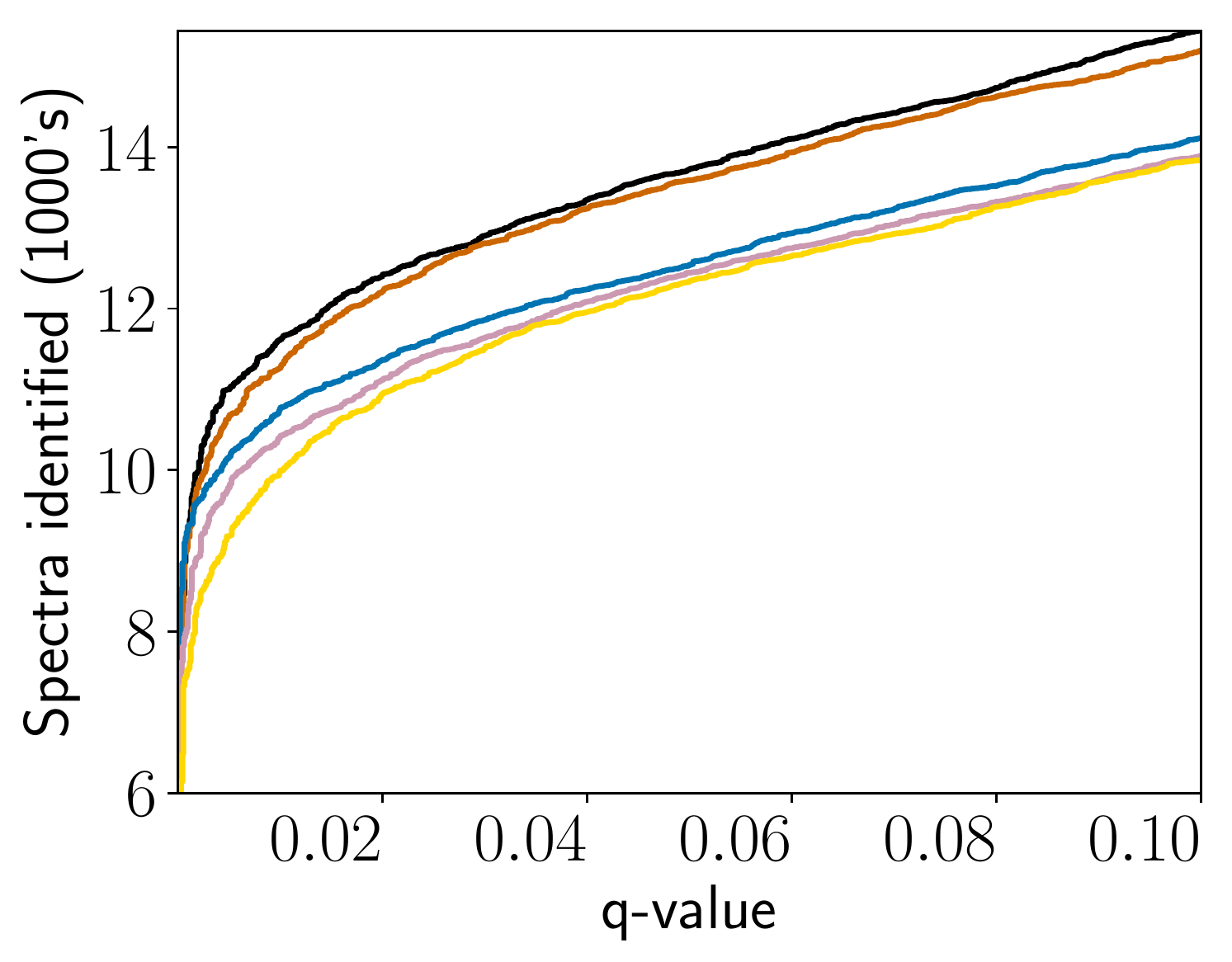}}
  \subfigure[Yeast-3]{\includegraphics[trim=0.45in 0.0in 0.0in 0.05in,
    clip=true,scale=0.3]{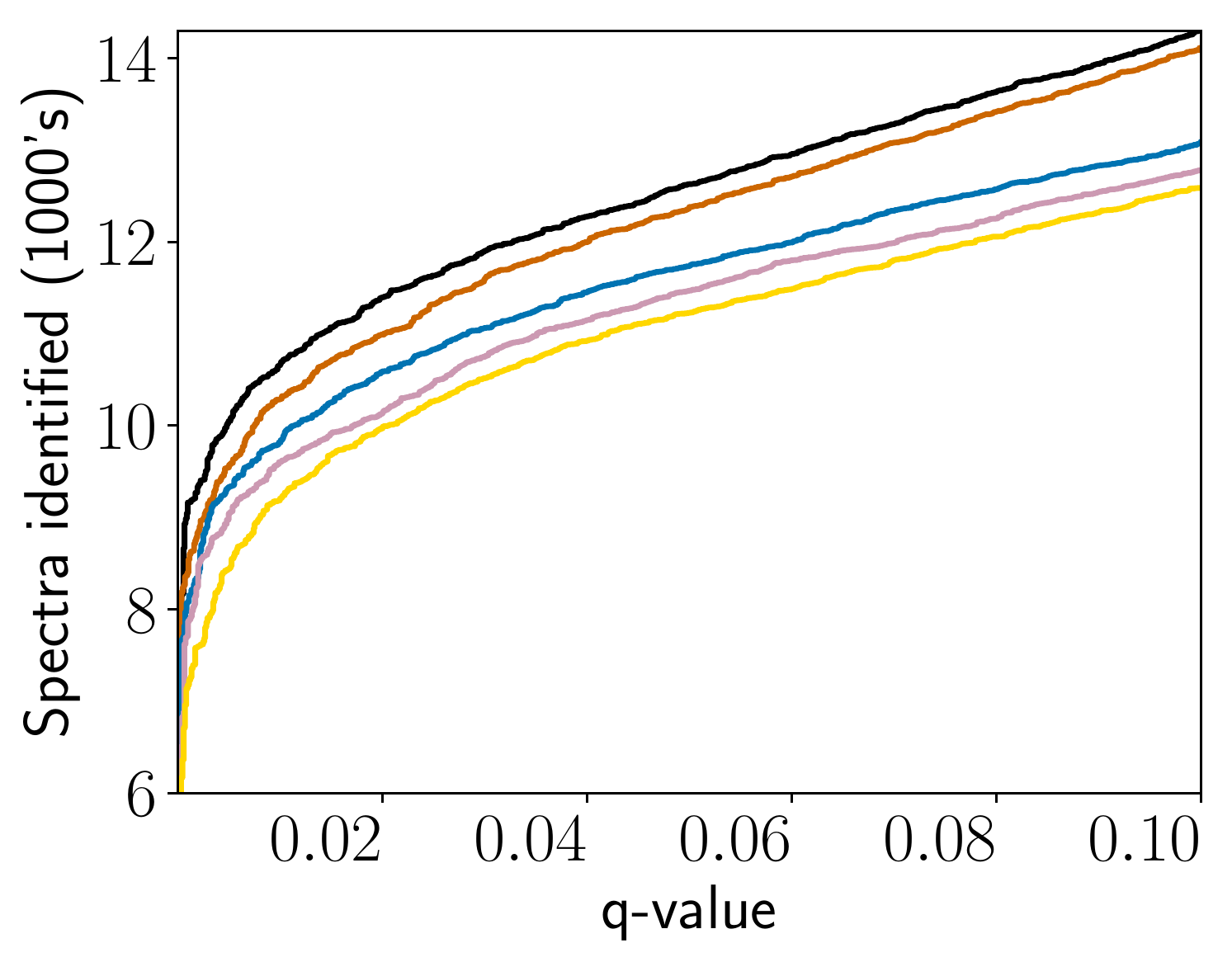}}
  \subfigure[Yeast-4]{\includegraphics[trim=0.45in 0.0in 0.0in 0.05in,
    clip=true,scale=0.3]{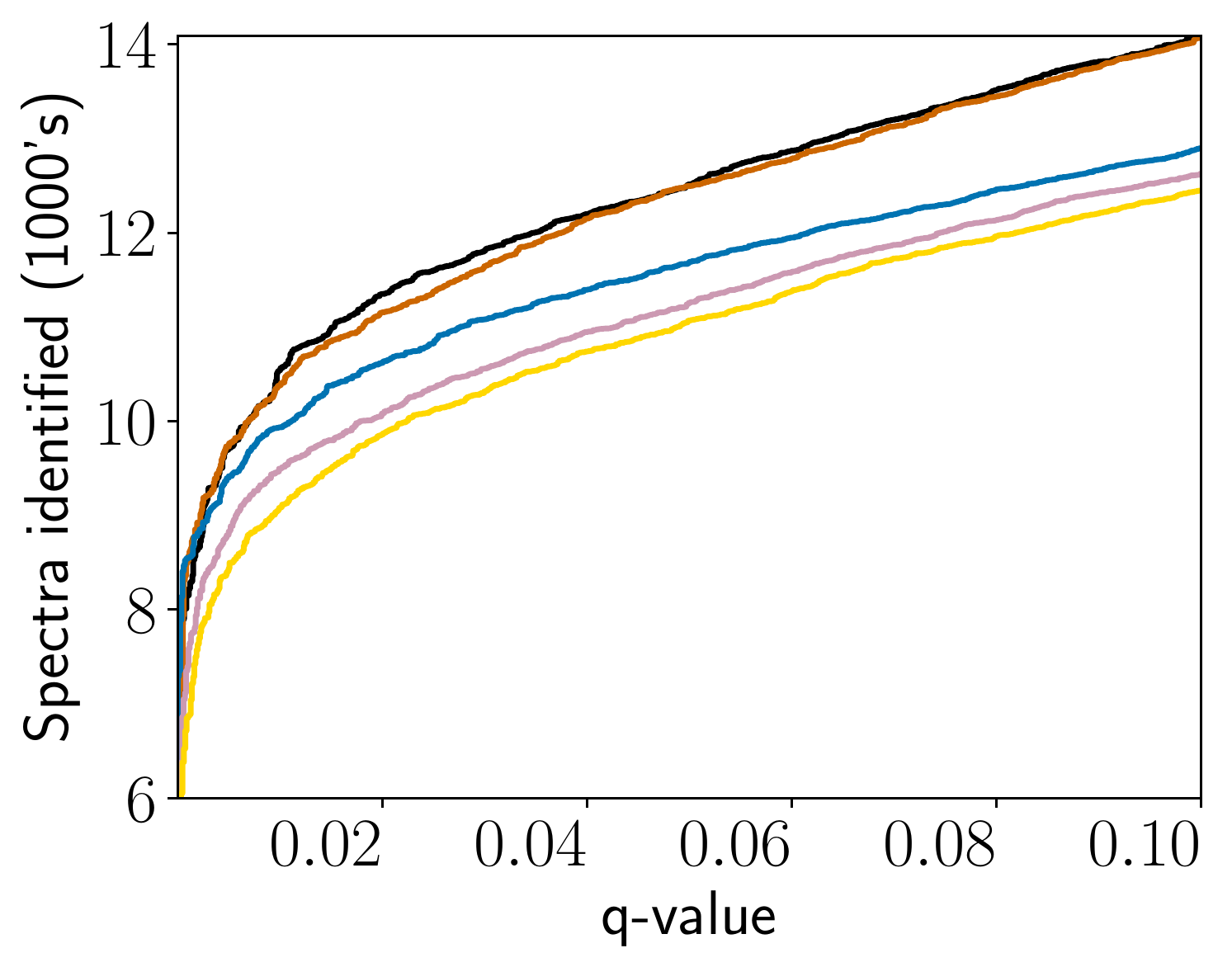}}
  \caption{{\small Post-database-search accuracy plots measured by
    $q$-value versus number of spectra identified for worm
    (\emph{C. elegans}) and yeast (\emph{Saccharomyces cerevisiae})
    datasets.  All methods are post-processed using the
    Percolator SVM classifier~\cite{kall:semi-supervised}.  ``DRIP
    Fisher''
    augments the standard set of DRIP PSM features (described
    in~\cite{halloran2016dynamic})  with the recently derived
    gradient-based DRIP features in~\cite{halloran2017gradients}.
    ``Didea Fisher'' uses the aforementioned DRIP features with the
    gradient features replaced by Didea's conditional log-likelihood
    gradients.
    ``XCorr,'' ``XCorr $p$-value,'' and ``MS-GF+'' use their standard
    sets of Percolator features (described
    in~\cite{halloran2016dynamic}).}
  }
  \label{fig:percolatorAbsRanking}
\end{figure*}

\end{document}